\documentclass[conference]{IEEEtran}
\usepackage{amsmath}
\usepackage{amssymb}
\usepackage{amsfonts}
\usepackage{mathrsfs}
\usepackage{xspace}
\usepackage{bm}
\usepackage{fancyref}
\usepackage{textcomp}
\usepackage[Symbol]{upgreek}

\newcommand{\safemath}[2]{\newcommand{#1}{\ensuremath{#2}\xspace}}
\renewcommand{\safemath}[2]{\newcommand{#1}{\ensuremath{#2}\xspace}}

\newcommand{\ssa}{\mathsf{a}}
\newcommand{\ssb}{\mathsf{b}}
\newcommand{\ssc}{\mathsf{c}}
\newcommand{\ssd}{\mathsf{d}}
\newcommand{\sse}{\mathsf{e}}
\newcommand{\ssf}{\mathsf{f}}
\newcommand{\ssg}{\mathsf{g}}
\newcommand{\ssh}{\mathsf{h}}
\newcommand{\ssi}{\mathsf{i}}
\newcommand{\ssj}{\mathsf{j}}
\newcommand{\ssk}{\mathsf{k}}
\newcommand{\ssl}{\mathsf{l}}
\newcommand{\ssm}{\mathsf{m}}
\newcommand{\ssn}{\mathsf{n}}
\newcommand{\sso}{\mathsf{o}}
\newcommand{\ssp}{\mathsf{p}}
\newcommand{\ssq}{\mathsf{q}}
\newcommand{\ssr}{\mathsf{r}}
\newcommand{\sss}{\mathsf{s}}
\newcommand{\sst}{\mathsf{t}}
\newcommand{\ssu}{\mathsf{u}}
\newcommand{\ssv}{\mathsf{v}}
\newcommand{\ssw}{\mathsf{w}}
\newcommand{\ssx}{\mathsf{x}}
\newcommand{\ssy}{\mathsf{y}}
\newcommand{\ssz}{\mathsf{z}}

\safemath{\bmsa}{\bm{\ssa}}
\safemath{\bmsb}{\bm{\ssb}}
\safemath{\bmsc}{\bm{\ssc}}
\safemath{\bmsd}{\bm{\ssd}}
\safemath{\bmse}{\bm{\sse}}
\safemath{\bmsf}{\bm{\ssf}}
\safemath{\bmsg}{\bm{\ssg}}
\safemath{\bmsh}{\bm{\ssh}}
\safemath{\bmsi}{\bm{\ssi}}
\safemath{\bmsj}{\bm{\ssj}}
\safemath{\bmsk}{\bm{\ssk}}
\safemath{\bmsl}{\bm{\ssl}}
\safemath{\bmsm}{\bm{\ssm}}
\safemath{\bmsn}{\bm{\ssn}}
\safemath{\bmso}{\bm{\sso}}
\safemath{\bmsp}{\bm{\ssp}}
\safemath{\bmsq}{\bm{\ssq}}
\safemath{\bmsr}{\bm{\ssr}}
\safemath{\bmss}{\bm{\sss}}
\safemath{\bmst}{\bm{\sst}}
\safemath{\bmsu}{\bm{\ssu}}
\safemath{\bmsv}{\bm{\ssv}}
\safemath{\bmsw}{\bm{\ssw}}
\safemath{\bmsx}{\bm{\ssx}}
\safemath{\bmsy}{\bm{\ssy}}
\safemath{\bmsz}{\bm{\ssz}}

\bmdefine{\bmualphad}{\upalpha}
\bmdefine{\bmubetad}{\upbeta}
\bmdefine{\bmuchid}{\upchi}
\bmdefine{\bmudeltad}{\updelta}
\bmdefine{\bmuepsilond}{\upepsilon}
\bmdefine{\bmuvarepsilond}{\upvarepsilon}
\bmdefine{\bmuetad}{\upeta}
\bmdefine{\bmugammad}{\upgamma}
\bmdefine{\bmuiotad}{\upiota}
\bmdefine{\bmukappad}{\upkappa}
\bmdefine{\bmulambdad}{\uplambda}
\bmdefine{\bmumud}{\upmu}
\bmdefine{\bmunud}{\upnu}
\bmdefine{\bmuomegad}{\upomega}
\bmdefine{\bmuphid}{\upphi}
\bmdefine{\bmuvarphid}{\upvarphi}
\bmdefine{\bmupid}{\uppi}
\bmdefine{\bmuvarpid}{\upvarpi}
\bmdefine{\bmupsid}{\uppsi}
\bmdefine{\bmurhod}{\uprho}
\bmdefine{\bmuvarrhod}{\upvarrho}
\bmdefine{\bmusigmad}{\upsigma}
\bmdefine{\bmuvarsigmad}{\upvarsigma}
\bmdefine{\bmutaud}{\uptau}
\bmdefine{\bmuthetad}{\uptheta}
\bmdefine{\bmuvarthetad}{\upvartheta}
\bmdefine{\bmuupsilond}{\upupsilon}
\bmdefine{\bmuxid}{\upxi}
\bmdefine{\bmuzetad}{\upzeta}

\safemath{\bmua}{\mathbf{a}}
\safemath{\bmub}{\mathbf{b}}
\safemath{\bmuc}{\mathbf{c}}
\safemath{\bmud}{\mathbf{d}}
\safemath{\bmue}{\mathbf{e}}
\safemath{\bmuf}{\mathbf{f}}
\safemath{\bmug}{\mathbf{g}}
\safemath{\bmuh}{\mathbf{h}}
\safemath{\bmui}{\mathbf{i}}
\safemath{\bmuj}{\mathbf{j}}
\safemath{\bmuk}{\mathbf{k}}
\safemath{\bmul}{\mathbf{l}}
\safemath{\bmum}{\mathbf{m}}
\safemath{\bmun}{\mathbf{n}}
\safemath{\bmuo}{\mathbf{o}}
\safemath{\bmup}{\mathbf{p}}
\safemath{\bmuq}{\mathbf{q}}
\safemath{\bmur}{\mathbf{r}}
\safemath{\bmus}{\mathbf{s}}
\safemath{\bmut}{\mathbf{t}}
\safemath{\bmuu}{\mathbf{u}}
\safemath{\bmuv}{\mathbf{v}}
\safemath{\bmuw}{\mathbf{w}}
\safemath{\bmux}{\mathbf{x}}
\safemath{\bmuy}{\mathbf{y}}
\safemath{\bmuz}{\mathbf{z}}

\safemath{\bmualpha}{\bmualphad}
\safemath{\bmubeta}{\bmubetad}
\safemath{\bmuchi}{\bumchid}
\safemath{\bmudelta}{\bmudeltad}
\safemath{\bmuepsilon}{\bmuepsilond}
\safemath{\bmuvarepsilon}{\bmuvarepsilond}
\safemath{\bmueta}{\bmuetad}
\safemath{\bmugamma}{\bmugammad}
\safemath{\bmuiota}{\bmuiotad}
\safemath{\bmukappa}{\bmukappad}
\safemath{\bmulambda}{\bmulambdad}
\safemath{\bmumu}{\bmumud}
\safemath{\bmunu}{\bmunud}
\safemath{\bmuomega}{\bmuomegad}
\safemath{\bmuphi}{\bmuphid}
\safemath{\bmuvarphi}{\bmuvarphid}
\safemath{\bmupi}{\bmupid}
\safemath{\bmuvarpi}{\bmuvarpid}
\safemath{\bmupsi}{\bmupsid}
\safemath{\bmurho}{\bmurhod}
\safemath{\bmuvarrho}{\bmuvarrhod}
\safemath{\bmusigma}{\bmusigmad}
\safemath{\bmuvarsigma}{\bmuvarsigmad}
\safemath{\bmutau}{\bmutaud}
\safemath{\bmutheta}{\bmuthetad}
\safemath{\bmuvartheta}{\bmuvarthetad}
\safemath{\bmuupsilon}{\bmuupsilond}
\safemath{\bmuxi}{\bmuxid}
\safemath{\bmuzeta}{\bmuzetad}

\bmdefine{\bmiad}{a}
\bmdefine{\bmibd}{b}
\bmdefine{\bmicd}{c}
\bmdefine{\bmidd}{d}
\bmdefine{\bmied}{e}
\bmdefine{\bmifd}{f}
\bmdefine{\bmigd}{g}
\bmdefine{\bmihd}{h}
\bmdefine{\bmiid}{i}
\bmdefine{\bmijd}{j}
\bmdefine{\bmikd}{k}
\bmdefine{\bmild}{l}
\bmdefine{\bmimd}{m}
\bmdefine{\bmind}{n}
\bmdefine{\bmiod}{o}
\bmdefine{\bmipd}{p}
\bmdefine{\bmiqd}{q}
\bmdefine{\bmird}{r}
\bmdefine{\bmisd}{s}
\bmdefine{\bmitd}{t}
\bmdefine{\bmiud}{u}
\bmdefine{\bmivd}{v}
\bmdefine{\bmiwd}{w}
\bmdefine{\bmixd}{x}
\bmdefine{\bmiyd}{y}
\bmdefine{\bmizd}{z}

\bmdefine{\bmialphad}{\alpha}
\bmdefine{\bmibetad}{\beta}
\bmdefine{\bmichid}{\chi}
\bmdefine{\bmideltad}{\delta}
\bmdefine{\bmiepsilond}{\epsilon}
\bmdefine{\bmivarepsilond}{\varepsilon}
\bmdefine{\bmietad}{\eta}
\bmdefine{\bmigammad}{\gamma}
\bmdefine{\bmiiotad}{\iota}
\bmdefine{\bmikappad}{\kappa}
\bmdefine{\bmivarkappad}{\varkappa}
\bmdefine{\bmilambdad}{\lambda}
\bmdefine{\bmimud}{\mu}
\bmdefine{\bminud}{\nu}
\bmdefine{\bmiomegad}{\omega}
\bmdefine{\bmiphid}{\phi}
\bmdefine{\bmivarphid}{\varphi}
\bmdefine{\bmipid}{\pi}
\bmdefine{\bmivarpid}{\varpi}
\bmdefine{\bmipsid}{\psi}
\bmdefine{\bmirhod}{\rho}
\bmdefine{\bmivarrhod}{\varrho}
\bmdefine{\bmisigmad}{\sigma}
\bmdefine{\bmivarsigmad}{\varsigma}
\bmdefine{\bmitaud}{\tau}
\bmdefine{\bmithetad}{\theta}
\bmdefine{\bmivarthetad}{\vartheta}
\bmdefine{\bmiupsilond}{\upsilon}
\bmdefine{\bmixid}{\xi}
\bmdefine{\bmizetad}{\zeta}

\safemath{\bmia}{\bmiad}
\safemath{\bmib}{\bmibd}
\safemath{\bmic}{\bmicd}
\safemath{\bmid}{\bmidd}
\safemath{\bmie}{\bmied}
\safemath{\bmif}{\bmifd}
\safemath{\bmig}{\bmigd}
\safemath{\bmih}{\bmihd}
\safemath{\bmii}{\bmiid}
\safemath{\bmij}{\bmijd}
\safemath{\bmik}{\bmikd}
\safemath{\bmil}{\bmild}
\safemath{\bmim}{\bmimd}
\safemath{\bmin}{\bmind}
\safemath{\bmio}{\bmiod}
\safemath{\bmip}{\bmipd}
\safemath{\bmiq}{\bmiqd}
\safemath{\bmir}{\bmird}
\safemath{\bmis}{\bmisd}
\safemath{\bmit}{\bmitd}
\safemath{\bmiu}{\bmiud}
\safemath{\bmiv}{\bmivd}
\safemath{\bmiw}{\bmiwd}
\safemath{\bmix}{\bmixd}
\safemath{\bmiy}{\bmiyd}
\safemath{\bmiz}{\bmizd}

\safemath{\bmialpha}{\bmialphad}
\safemath{\bmibeta}{\bmibetad}
\safemath{\bmichi}{\bmichid}
\safemath{\bmidelta}{\bmideltad}
\safemath{\bmiepsilon}{\bmiepsilond}
\safemath{\bmivarepsilon}{\bmivarepsilond}
\safemath{\bmieta}{\bmietad}
\safemath{\bmigamma}{\bmigammad}
\safemath{\bmiiota}{\bmiiotad}
\safemath{\bmikappa}{\bmikappad}
\safemath{\bmivarkappa}{\bmivarkappad}
\safemath{\bmilambda}{\bmilambdad}
\safemath{\bmimu}{\bmimud}
\safemath{\bminu}{\bminud}
\safemath{\bmiomega}{\bmiomegad}
\safemath{\bmiphi}{\bmiphid}
\safemath{\bmivarphi}{\bmivarphid}
\safemath{\bmipi}{\bmipid}
\safemath{\bmivarpi}{\bmivarpid}
\safemath{\bmipsi}{\bmipsid}
\safemath{\bmirho}{\bmirhod}
\safemath{\bmivarrho}{\bmivarrhod}
\safemath{\bmisigma}{\bmisigmad}
\safemath{\bmivarsigma}{\bmivarsigmad}
\safemath{\bmitau}{\bmitaud}
\safemath{\bmitheta}{\bmithetad}
\safemath{\bmivartheta}{\bmivarthetad}
\safemath{\bmiupsilon}{\bmiupsilond}
\safemath{\bmixi}{\bmixid}
\safemath{\bmizeta}{\bmizetad}

\bmdefine{\bmuDeltad}{\Updelta}
\bmdefine{\bmuGammad}{\Upgamma}
\bmdefine{\bmuLambdad}{\Uplambda}
\bmdefine{\bmuOmegad}{\Upomega}
\bmdefine{\bmuPhid}{\Upphi}
\bmdefine{\bmuPid}{\Uppi}
\bmdefine{\bmuPsid}{\Uppsi}
\bmdefine{\bmuSigmad}{\Upsigma}
\bmdefine{\bmuThetad}{\Uptheta}
\bmdefine{\bmuUpsilond}{\Upupsilon}
\bmdefine{\bmuXid}{\Upxi}

\safemath{\bmuA}{\mathbf{A}}
\safemath{\bmuB}{\mathbf{B}}
\safemath{\bmuC}{\mathbf{C}}
\safemath{\bmuD}{\mathbf{D}}
\safemath{\bmuE}{\mathbf{E}}
\safemath{\bmuF}{\mathbf{F}}
\safemath{\bmuG}{\mathbf{G}}
\safemath{\bmuH}{\mathbf{H}}
\safemath{\bmuI}{\mathbf{I}}
\safemath{\bmuJ}{\mathbf{J}}
\safemath{\bmuK}{\mathbf{K}}
\safemath{\bmuL}{\mathbf{L}}
\safemath{\bmuM}{\mathbf{M}}
\safemath{\bmuN}{\mathbf{N}}
\safemath{\bmuO}{\mathbf{O}}
\safemath{\bmuP}{\mathbf{P}}
\safemath{\bmuQ}{\mathbf{Q}}
\safemath{\bmuR}{\mathbf{R}}
\safemath{\bmuS}{\mathbf{S}}
\safemath{\bmuT}{\mathbf{T}}
\safemath{\bmuU}{\mathbf{U}}
\safemath{\bmuV}{\mathbf{V}}
\safemath{\bmuW}{\mathbf{W}}
\safemath{\bmuX}{\mathbf{X}}
\safemath{\bmuY}{\mathbf{Y}}
\safemath{\bmuZ}{\mathbf{Z}}

\safemath{\bmuZero}{\mathbf{0}}
\safemath{\bmuOne}{\mathbf{1}}

\safemath{\bmuDelta}{\bmuDeltad}
\safemath{\bmuGamma}{\bmuGammad}
\safemath{\bmuLambda}{\bmuLambdad}
\safemath{\bmuOmega}{\bmuOmegad}
\safemath{\bmuPhi}{\bmuPhid}
\safemath{\bmuPi}{\bmuPid}
\safemath{\bmuPsi}{\bmuPsid}
\safemath{\bmuSigma}{\bmuSigmad}
\safemath{\bmuTheta}{\bmuThetad}
\safemath{\bmuUpsilon}{\bmuUpsilond}
\safemath{\bmuXi}{\bmuXid}

\bmdefine{\bmiAd}{A}
\bmdefine{\bmiBd}{B}
\bmdefine{\bmiCd}{C}
\bmdefine{\bmiDd}{D}
\bmdefine{\bmiEd}{E}
\bmdefine{\bmiFd}{F}
\bmdefine{\bmiGd}{G}
\bmdefine{\bmiHd}{H}
\bmdefine{\bmiId}{I}
\bmdefine{\bmiJd}{J}
\bmdefine{\bmiKd}{K}
\bmdefine{\bmiLd}{L}
\bmdefine{\bmiMd}{M}
\bmdefine{\bmiOd}{N}
\bmdefine{\bmiPd}{O}
\bmdefine{\bmiQd}{P}
\bmdefine{\bmiRd}{R}
\bmdefine{\bmiSd}{S}
\bmdefine{\bmiTd}{T}
\bmdefine{\bmiUd}{U}
\bmdefine{\bmiVd}{V}
\bmdefine{\bmiWd}{W}
\bmdefine{\bmiXd}{X}
\bmdefine{\bmiYd}{Y}
\bmdefine{\bmiZd}{Z}

\bmdefine{\bmiDeltad}{\Delta}
\bmdefine{\bmiGammad}{\Gamma}
\bmdefine{\bmiLambdad}{\Lambda}
\bmdefine{\bmiOmegad}{\Omega}
\bmdefine{\bmiPhid}{\Phi}
\bmdefine{\bmiPid}{\Pi}
\bmdefine{\bmiPsid}{\Psi}
\bmdefine{\bmiSigmad}{\Sigma}
\bmdefine{\bmiThetad}{\Theta}
\bmdefine{\bmiUpsilond}{\Upsilon}
\bmdefine{\bmiXid}{\Xi}

\safemath{\bmiA}{\bmiAd}
\safemath{\bmiB}{\bmiBd}
\safemath{\bmiC}{\bmiCd}
\safemath{\bmiD}{\bmiDd}
\safemath{\bmiE}{\bmiEd}
\safemath{\bmiF}{\bmiFd}
\safemath{\bmiG}{\bmiGd}
\safemath{\bmiH}{\bmiHd}
\safemath{\bmiI}{\bmiId}
\safemath{\bmiJ}{\bmiJd}
\safemath{\bmiK}{\bmiKd}
\safemath{\bmiL}{\bmiLd}
\safemath{\bmiM}{\bmiMd}
\safemath{\bmiN}{\bmiNd}
\safemath{\bmiO}{\bmiOd}
\safemath{\bmiP}{\bmiPd}
\safemath{\bmiQ}{\bmiQd}
\safemath{\bmiR}{\bmiRd}
\safemath{\bmiS}{\bmiSd}
\safemath{\bmiT}{\bmiTd}
\safemath{\bmiU}{\bmiUd}
\safemath{\bmiV}{\bmiVd}
\safemath{\bmiW}{\bmiWd}
\safemath{\bmiX}{\bmiXd}
\safemath{\bmiY}{\bmiYd}
\safemath{\bmiZ}{\bmiZd}

\safemath{\bmiDelta}{\bmiDeltad}
\safemath{\bmiGamma}{\bmiGammad}
\safemath{\bmiLambda}{\bmiLambdad}
\safemath{\bmiOmega}{\bmiOmegad}
\safemath{\bmiPhi}{\bmiPhid}
\safemath{\bmiPi}{\bmiPid}
\safemath{\bmiPsi}{\bmiPsid}
\safemath{\bmiSigma}{\bmiSigmad}
\safemath{\bmiTheta}{\bmiThetad}
\safemath{\bmiUpsilon}{\bmiUpsilond}
\safemath{\bmiXi}{\bmiXid}

\safemath{\setA}{\mathcal{A}}
\safemath{\setB}{\mathcal{B}}
\safemath{\setC}{\mathcal{C}}
\safemath{\setD}{\mathcal{D}}
\safemath{\setE}{\mathcal{E}}
\safemath{\setF}{\mathcal{F}}
\safemath{\setG}{\mathcal{G}}
\safemath{\setH}{\mathcal{H}}
\safemath{\setI}{\mathcal{I}}
\safemath{\setJ}{\mathcal{J}}
\safemath{\setK}{\mathcal{K}}
\safemath{\setL}{\mathcal{L}}
\safemath{\setM}{\mathcal{M}}
\safemath{\setN}{\mathcal{N}}
\safemath{\setO}{\mathcal{O}}
\safemath{\setP}{\mathcal{P}}
\safemath{\setQ}{\mathcal{Q}}
\safemath{\setR}{\mathcal{R}}
\safemath{\setS}{\mathcal{S}}
\safemath{\setT}{\mathcal{T}}
\safemath{\setU}{\mathcal{U}}
\safemath{\setV}{\mathcal{V}}
\safemath{\setW}{\mathcal{W}}
\safemath{\setX}{\mathcal{X}}
\safemath{\setY}{\mathcal{Y}}
\safemath{\setZ}{\mathcal{Z}}
\safemath{\emptySet}{\varnothing}

\safemath{\colA}{\mathscr{A}}
\safemath{\colB}{\mathscr{B}}
\safemath{\colC}{\mathscr{C}}
\safemath{\colD}{\mathscr{D}}
\safemath{\colE}{\mathscr{E}}
\safemath{\colF}{\mathscr{F}}
\safemath{\colG}{\mathscr{G}}
\safemath{\colH}{\mathscr{H}}
\safemath{\colI}{\mathscr{I}}
\safemath{\colJ}{\mathscr{J}}
\safemath{\colK}{\mathscr{K}}
\safemath{\colL}{\mathscr{L}}
\safemath{\colM}{\mathscr{M}}
\safemath{\colN}{\mathscr{N}}
\safemath{\colO}{\mathscr{O}}
\safemath{\colP}{\mathscr{P}}
\safemath{\colQ}{\mathscr{Q}}
\safemath{\colR}{\mathscr{R}}
\safemath{\colS}{\mathscr{S}}
\safemath{\colT}{\mathscr{T}}
\safemath{\colU}{\mathscr{U}}
\safemath{\colV}{\mathscr{V}}
\safemath{\colW}{\mathscr{W}}
\safemath{\colX}{\mathscr{X}}
\safemath{\colY}{\mathscr{Y}}
\safemath{\colZ}{\mathscr{Z}}

\safemath{\opA}{\mathbb{A}}
\safemath{\opB}{\mathbb{B}}
\safemath{\opC}{\mathbb{C}}
\safemath{\opD}{\mathbb{D}}
\safemath{\opE}{\mathbb{E}}
\safemath{\opF}{\mathbb{F}}
\safemath{\opG}{\mathbb{G}}
\safemath{\opH}{\mathbb{H}}
\safemath{\opI}{\mathbb{I}}
\safemath{\opJ}{\mathbb{J}}
\safemath{\opK}{\mathbb{K}}
\safemath{\opL}{\mathbb{L}}
\safemath{\opM}{\mathbb{M}}
\safemath{\opN}{\mathbb{N}}
\safemath{\opO}{\mathbb{O}}
\safemath{\opP}{\mathbb{P}}
\safemath{\opQ}{\mathbb{Q}}
\safemath{\opR}{\mathbb{R}}
\safemath{\opS}{\mathbb{S}}
\safemath{\opT}{\mathbb{T}}
\safemath{\opU}{\mathbb{U}}
\safemath{\opV}{\mathbb{V}}
\safemath{\opW}{\mathbb{W}}
\safemath{\opX}{\mathbb{X}}
\safemath{\opY}{\mathbb{Y}}
\safemath{\opZ}{\mathbb{Z}}
\safemath{\opZero}{\mathbb{O}}
\safemath{\identityop}{\opI}

\safemath{\sca}{a}
\safemath{\scb}{b}
\safemath{\scc}{c}
\safemath{\scd}{d}
\safemath{\sce}{e}
\safemath{\scf}{f}
\safemath{\scg}{g}
\safemath{\sch}{h}
\safemath{\sci}{i}
\safemath{\scj}{j}
\safemath{\sck}{k}
\safemath{\scl}{l}
\safemath{\scm}{m}
\safemath{\scn}{n}
\safemath{\sco}{o}
\safemath{\scp}{p}
\safemath{\scq}{q}
\safemath{\scr}{r}
\safemath{\scs}{s}
\safemath{\sct}{t}
\safemath{\scu}{u}
\safemath{\scv}{v}
\safemath{\scw}{w}
\safemath{\scx}{x}
\safemath{\scy}{y}
\safemath{\scz}{z}

\safemath{\scA}{A}
\safemath{\scB}{B}
\safemath{\scC}{C}
\safemath{\scD}{D}
\safemath{\scE}{E}
\safemath{\scF}{F}
\safemath{\scG}{G}
\safemath{\scH}{H}
\safemath{\scI}{I}
\safemath{\scJ}{J}
\safemath{\scK}{K}
\safemath{\scL}{L}
\safemath{\scM}{M}
\safemath{\scN}{N}
\safemath{\scO}{O}
\safemath{\scP}{P}
\safemath{\scQ}{Q}
\safemath{\scR}{R}
\safemath{\scS}{S}
\safemath{\scT}{T}
\safemath{\scU}{U}
\safemath{\scV}{V}
\safemath{\scW}{W}
\safemath{\scX}{X}
\safemath{\scY}{Y}
\safemath{\scZ}{Z}

\safemath{\scalpha}{\alpha}
\safemath{\scbeta}{\beta}
\safemath{\scchi}{\chi}
\safemath{\scdelta}{\delta}
\safemath{\scepsilon}{\epsilon}
\safemath{\scvarepsilon}{\varepsilon}
\safemath{\sceta}{\eta}
\safemath{\scgamma}{\gamma}
\safemath{\sciota}{\iota}
\safemath{\sckappa}{\kappa}
\safemath{\scvarkappa}{\varkappa}
\safemath{\sclambda}{\lambda}
\safemath{\scmu}{\mu}
\safemath{\scnu}{\nu}
\safemath{\scomega}{\omega}
\safemath{\scphi}{\phi}
\safemath{\scvarphi}{\varphi}
\safemath{\scpi}{\pi}
\safemath{\scvarpi}{\varpi}
\safemath{\scpsi}{\psi}
\safemath{\scrho}{\rho}
\safemath{\scvarrho}{\varrho}
\safemath{\scsigma}{\sigma}
\safemath{\scvarsigma}{\varsigma}
\safemath{\sctau}{\tau}
\safemath{\sctheta}{\theta}
\safemath{\scvartheta}{\vartheta}
\safemath{\scupsilon}{\upsilon}
\safemath{\scxi}{\xi}
\safemath{\sczeta}{\zeta}

\safemath{\veca}{\mathbf{a}}
\safemath{\vecb}{\mathbf{b}}
\safemath{\vecc}{\mathbf{c}}
\safemath{\vecd}{\mathbf{d}}
\safemath{\vece}{\mathbf{e}}
\safemath{\vecf}{\mathbf{f}}
\safemath{\vecg}{\mathbf{g}}
\safemath{\vech}{\mathbf{h}}
\safemath{\veci}{\mathbf{i}}
\safemath{\vecj}{\mathbf{j}}
\safemath{\veck}{\mathbf{k}}
\safemath{\vecl}{\mathbf{l}}
\safemath{\vecm}{\mathbf{m}}
\safemath{\vecn}{\mathbf{n}}
\safemath{\veco}{\mathbf{o}}
\safemath{\vecp}{\mathbf{p}}
\safemath{\vecq}{\mathbf{q}}
\safemath{\vecr}{\mathbf{r}}
\safemath{\vecs}{\mathbf{s}}
\safemath{\vect}{\mathbf{t}}
\safemath{\vecu}{\mathbf{u}}
\safemath{\vectU}{\mathbf{U}}
\safemath{\vecv}{\mathbf{v}}
\safemath{\vecw}{\mathbf{w}}
\safemath{\vecx}{\mathbf{x}}
\safemath{\vectX}{\mathbf{X}}
\safemath{\vecy}{\mathbf{y}}
\safemath{\vecz}{\mathbf{z}}
\safemath{\veczero}{\mathbf{0}}
\safemath{\vecone}{\mathbf{1}}

\safemath{\vecalpha}{\upalpha}
\safemath{\vecbeta}{\upbeta}
\safemath{\vecchi}{\upchi}
\safemath{\vecdelta}{\updelta}
\safemath{\vecepsilon}{\upepsilon}
\safemath{\vecvarepsilon}{\upvarepsilon}
\safemath{\veceta}{\upeta}
\safemath{\vecgamma}{\upgamma}
\safemath{\veciota}{\upiota}
\safemath{\veckappa}{\upkappa}
\safemath{\veclambda}{\uplambda}
\safemath{\vecmu}{\text{\textmu}}
\safemath{\vecnu}{\upnu}
\safemath{\vecomega}{\upomega}
\safemath{\vecphi}{\upphi}
\safemath{\vecvarphi}{\upvarphi}
\safemath{\vecpi}{\uppi}
\safemath{\vecvarpi}{\upvarpi}
\safemath{\vecpsi}{\uppsi}
\safemath{\vecrho}{\uprho}
\safemath{\vecvarrho}{\upvarrho}
\safemath{\vecsigma}{\upsigma}
\safemath{\vecvarsigma}{\upvarsigma}
\safemath{\vectau}{\uptau}
\safemath{\vectheta}{\uptheta}
\safemath{\vecvartheta}{\upvartheta}
\safemath{\vecupsilon}{\upupsilon}
\safemath{\vecxi}{\upxi}
\safemath{\veczeta}{\upzeta}

\safemath{\matA}{\mathrm{A}}
\safemath{\matB}{\mathrm{B}}
\safemath{\matC}{\mathrm{C}}
\safemath{\matD}{\mathrm{D}}
\safemath{\matE}{\mathrm{E}}
\safemath{\matF}{\mathrm{F}}
\safemath{\matG}{\mathrm{G}}
\safemath{\matH}{\mathrm{H}}
\safemath{\matI}{\mathrm{I}}
\safemath{\matJ}{\mathrm{J}}
\safemath{\matK}{\mathrm{K}}
\safemath{\matL}{\mathrm{L}}
\safemath{\matM}{\mathrm{M}}
\safemath{\matN}{\mathrm{N}}
\safemath{\matO}{\mathrm{O}}
\safemath{\matP}{\mathrm{P}}
\safemath{\matQ}{\mathrm{Q}}
\safemath{\matR}{\mathrm{R}}
\safemath{\matS}{\mathrm{S}}
\safemath{\matT}{\mathrm{T}}
\safemath{\matU}{\mathrm{U}}
\safemath{\matV}{\mathrm{V}}
\safemath{\matW}{\mathrm{W}}
\safemath{\matX}{\mathrm{X}}
\safemath{\matY}{\mathrm{Y}}
\safemath{\matZ}{\mathrm{Z}}
\safemath{\matzero}{\mathrm{0}}

\safemath{\matDelta}{\Updelta}
\safemath{\matGamma}{\Upgammma}
\safemath{\matLambda}{\Uplambda}
\safemath{\matOmega}{\Upomega}
\safemath{\matPhi}{\Upphi}
\safemath{\matPi}{\Uppi}
\safemath{\matPsi}{\Uppsi}
\safemath{\matSigma}{\Upsigma}
\safemath{\matTheta}{\Uptheta}
\safemath{\matUpsilon}{\Upupsilon}
\safemath{\matXi}{\Upxi}

\safemath{\matidentity}{\matI}
\safemath{\matone}{\matO}

\safemath{\rnda}{\bmia}
\safemath{\rndb}{\bmib}
\safemath{\rndc}{\bmic}
\safemath{\rndd}{\bmid}
\safemath{\rnde}{\bmie}
\safemath{\rndf}{\bmif}
\safemath{\rndg}{\bmig}
\safemath{\rndh}{\bmih}
\safemath{\rndi}{\bmii}
\safemath{\rndj}{\bmij}
\safemath{\rndk}{\bmik}
\safemath{\rndl}{\bmil}
\safemath{\rndm}{\bmim}
\safemath{\rndn}{\bmin}
\safemath{\rndo}{\bmio}
\safemath{\rndp}{\bmip}
\safemath{\rndq}{\bmiq}
\safemath{\rndr}{\bmir}
\safemath{\rnds}{\bmis}
\safemath{\rndt}{\bmit}
\safemath{\rndu}{\bmiu}
\safemath{\rndv}{\bmiv}
\safemath{\rndw}{\bmiw}
\safemath{\rndx}{\bmix}
\safemath{\rndy}{\bmiy}
\safemath{\rndz}{\bmiz}

\safemath{\rndA}{\scA}
\safemath{\rndB}{\scB}
\safemath{\rndC}{\scC}
\safemath{\rndD}{\scD}
\safemath{\rndE}{\scE}
\safemath{\rndF}{\scF}
\safemath{\rndG}{\scG}
\safemath{\rndH}{\scH}
\safemath{\rndI}{\scI}
\safemath{\rndJ}{\scJ}
\safemath{\rndK}{\scK}
\safemath{\rndL}{\scL}
\safemath{\rndM}{\scM}
\safemath{\rndN}{\scN}
\safemath{\rndO}{\scO}
\safemath{\rndP}{\scP}
\safemath{\rndQ}{\scQ}
\safemath{\rndR}{\scR}
\safemath{\rndS}{\scS}
\safemath{\rndT}{\scT}
\safemath{\rndU}{\scU}
\safemath{\rndV}{\scV}
\safemath{\rndW}{\scW}
\safemath{\rndX}{\scX}
\safemath{\rndY}{\scY}
\safemath{\rndZ}{\scZ}

\safemath{\rndalpha}{\bmialpha}
\safemath{\rndbeta}{\bmibeta}
\safemath{\rndchi}{\bmichi}
\safemath{\rnddelta}{\bmidelta}
\safemath{\rndepsilon}{\bmiepsilon}
\safemath{\rndvarepsilon}{\bmivarepsilon}
\safemath{\rndeta}{\bmieta}
\safemath{\rndgamma}{\bmigamma}
\safemath{\rndiota}{\bmiiota}
\safemath{\rndkappa}{\bmikappa}
\safemath{\rndlambda}{\bmilambda}
\safemath{\rndmu}{\bmimu}
\safemath{\rndnu}{\bminu}
\safemath{\rndomega}{\bmiomega}
\safemath{\rndphi}{\bmiphi}
\safemath{\rndvarphi}{\bmivarphi}
\safemath{\rndpi}{\bmipi}
\safemath{\rndvarpi}{\bmivarpi}
\safemath{\rndpsi}{\bmipsi}
\safemath{\rndrho}{\bmirho}
\safemath{\rndvarrho}{\bmivarrho}
\safemath{\rndsigma}{\bmisigma}
\safemath{\rndvarsigma}{\bmivarsigma}
\safemath{\rndtau}{\bmitau}
\safemath{\rndtheta}{\bmitheta}
\safemath{\rndvartheta}{\bmivartheta}
\safemath{\rndupsilon}{\bmiupsilon}
\safemath{\rndxi}{\bmixi}
\safemath{\rndzeta}{\bmizeta}

\safemath{\rveca}{\bmua}
\safemath{\rvecb}{\bmub}
\safemath{\rvecc}{\bmuc}
\safemath{\rvecd}{\bmud}
\safemath{\rvece}{\bmue}
\safemath{\rvecf}{\bmuf}
\safemath{\rvecg}{\bmug}
\safemath{\rvech}{\bmuh}
\safemath{\rveci}{\bmui}
\safemath{\rvecj}{\bmuj}
\safemath{\rveck}{\bmuk}
\safemath{\rvecl}{\bmul}
\safemath{\rvecm}{\bmum}
\safemath{\rvecn}{\bmun}
\safemath{\rveco}{\bmuo}
\safemath{\rvecp}{\bmup}
\safemath{\rvecq}{\bmuq}
\safemath{\rvecr}{\bmur}
\safemath{\rvecs}{\bmus}
\safemath{\rvect}{\bmut}
\safemath{\rvecu}{\bmuu}
\safemath{\rvecv}{\bmuv}
\safemath{\rvecw}{\bmuw}
\safemath{\rvecx}{\bmux}
\safemath{\rvecy}{\bmuy}
\safemath{\rvecz}{\bmuz}

\safemath{\rvecalpha}{\bmualpha}
\safemath{\rvecbeta}{\bmubeta}
\safemath{\rvecchi}{\bmuchi}
\safemath{\rvecdelta}{\bmudelta}
\safemath{\rvecepsilon}{\bmuepsilon}
\safemath{\rvecvarepsilon}{\bmuvarepsilon}
\safemath{\rveceta}{\bmueta}
\safemath{\rvecgamma}{\bmugamma}
\safemath{\rveciota}{\bmuiota}
\safemath{\rveckappa}{\bmukappa}
\safemath{\rveclambda}{\bmulambda}
\safemath{\rvecmu}{\bmumu}
\safemath{\rvecnu}{\bmunu}
\safemath{\rvecomega}{\bmuomega}
\safemath{\rvecphi}{\bmuphi}
\safemath{\rvecvarphi}{\bmuvarphi}
\safemath{\rvecpi}{\bmupi}
\safemath{\rvecvarpi}{\bmuvarpi}
\safemath{\rvecpsi}{\bmupsi}
\safemath{\rvecrho}{\bmurho}
\safemath{\rvecvarrho}{\bmuvarrho}
\safemath{\rvecsigma}{\bmusigma}
\safemath{\rvecvarsigma}{\bmuvarsigma}
\safemath{\rvectau}{\bmutau}
\safemath{\rvectheta}{\bmutheta}
\safemath{\rvecvartheta}{\bmuvartheta}
\safemath{\rvecupsilon}{\bmuupsilon}
\safemath{\rvecxi}{\bmuxi}
\safemath{\rveczeta}{\bmuzeta}

\safemath{\rmatA}{\bmuA}
\safemath{\rmatB}{\bmuB}
\safemath{\rmatC}{\bmuC}
\safemath{\rmatD}{\bmuD}
\safemath{\rmatE}{\bmuE}
\safemath{\rmatF}{\bmuF}
\safemath{\rmatG}{\bmuG}
\safemath{\rmatH}{\bmuH}
\safemath{\rmatI}{\bmuI}
\safemath{\rmatJ}{\bmuJ}
\safemath{\rmatK}{\bmuK}
\safemath{\rmatL}{\bmuL}
\safemath{\rmatM}{\bmuM}
\safemath{\rmatN}{\bmuN}
\safemath{\rmatO}{\bmuO}
\safemath{\rmatP}{\bmuP}
\safemath{\rmatQ}{\bmuQ}
\safemath{\rmatR}{\bmuR}
\safemath{\rmatS}{\bmuS}
\safemath{\rmatT}{\bmuT}
\safemath{\rmatU}{\bmuU}
\safemath{\rmatV}{\bmuV}
\safemath{\rmatW}{\bmuW}
\safemath{\rmatX}{\bmuX}
\safemath{\rmatY}{\bmuY}
\safemath{\rmatZ}{\bmuZ}

\safemath{\rmatDelta}{\bmuDelta}
\safemath{\rmatGamma}{\bmuGamma}
\safemath{\rmatLambda}{\bmuLambda}
\safemath{\rmatOmega}{\bmuOmega}
\safemath{\rmatPhi}{\bmuPhi}
\safemath{\rmatPi}{\bmuPi}
\safemath{\rmatPsi}{\bmuPsi}
\safemath{\rmatSigma}{\bmuSigma}
\safemath{\rmatTheta}{\bmuTheta}
\safemath{\rmatUpsilon}{\bmuUpsilon}
\safemath{\rmatXi}{\bmuXi}

\safemath{\rndvecU}{\bmiU} 
\safemath{\rndvecW}{\bmiW} 
\safemath{\rndvecV}{\bmiV} 
\safemath{\rndvecX}{\bmiX} 
\safemath{\rndvecY}{\bmiY} 
\safemath{\rndvecC}{\bmiC} 

\newenvironment{textbmatrix}{	\setlength{\arraycolsep}{2.5pt}%
								\big[\begin{matrix}}{\end{matrix}\big]%
								\raisebox{0.08ex}{\vphantom{M}}}

\def\be{\begin{equation}}
\def\ee{\end{equation}}
\def\een{\nonumber \end{equation}}
\def\mat{\begin{bmatrix}}
\def\emat{\end{bmatrix}}
\def\btm{\begin{textbmatrix}}
\def\etm{\end{textbmatrix}}

\def\ba#1\ea{\begin{align}#1\end{align}}
\def\bas#1\eas{\begin{align*}#1\end{align*}}
\def\bs#1\es{\begin{split}#1\end{split}} 
\def\bg#1\eg{\begin{gather}#1\end{gather}}
\def\bml#1\eml{\begin{multline}#1\end{multline}}
\def\bi#1\ei{\begin{itemize}#1\end{itemize}} 
\def\bipi#1\eipi{\begin{inparaitem}#1\end{inparaitem}}

\newcommand{\lefto}{\mathopen{}\left}

\DeclareMathOperator{\adj}{adj}				
\DeclareMathOperator*{\argmin}{arg\;min}		
\DeclareMathOperator{\Exop}{\opE}			

\safemath{\fun}{\scf}						
\safemath{\altfun}{\scg}
\safemath{\aaltfun}{\sch}
\safemath{\bel}{\sce}					
\safemath{\altbel}{\sce}					
\safemath{\frel}{g}					
\safemath{\altfrel}{g}					
\safemath{\dfrel}{\tilde{g}}					
\safemath{\altdfrel}{\tilde{g}}					

\newcommand{\nullspace}{\setN}	 			
\newcommand{\Ex}[2]{\ensuremath{\Exop_{#1}\lefto[#2\right]}} 	
\newcommand{\bigEx}[2]{\ensuremath{\Exop_{#1} \! \bigl[#2\bigr]}} 	
\newcommand{\card}[1]{\left\lvert#1\right\rvert}			

\newcommand{\conj}[1]{\ensuremath{\overline{#1}}} 	
\newcommand{\inv}[1]{\ensuremath{#1^{-1}}} 	

\safemath{\dirac}{\delta}					
\safemath{\diracp}{\dirac(\time)}			
\safemath{\krond}{\dirac}					

\newcommand{\set}[1]{\left\{#1\right\}}		
\safemath{\upto}{\uparrow}
\safemath{\downto}{\downarrow}
\safemath{\iu}{i}							
\safemath{\maj}{\succ}
\newcommand{\dftmat}[1]{\matF_{#1}}			
\safemath{\mdft}{\dftmat{}}					
\safemath{\runity}{\beta}					
\safemath{\eval}{\lambda}					
\safemath{\veval}{\veclambda}				
\safemath{\littleo}{\sco}					

\let\im\undefined
\safemath{\re}{\mathfrak{Re}}				
\safemath{\im}{\mathfrak{Im}}				

\safemath{\euclidspace}{\complexset}			
\safemath{\confunspace}{\setC}				
\newcommand{\banachseqspace}[1]{l^{#1}}		
\safemath{\hilseqspace}{\banachseqspace{2}}	
\newcommand{\banachfunspace}[1]{\setL^{#1}}	
\safemath{\hilfunspace}{\banachfunspace{2}}	
\safemath{\schwarzspace}{\setS}				

\newcommand{\hadj}[1]{#1^{\star}}			

\safemath{\SNR}{\text{\sc snr}} 				
\safemath{\No}{N_0}							
\safemath{\Es}{E_s}							
\safemath{\Eb}{E_b}							
\safemath{\EbNo}{\frac{\Eb}{\No}}
\safemath{\EsNo}{\frac{\Es}{\No}}

\let\time\undefined
\safemath{\time}{\sct}						
\safemath{\dtime}{\sck}						
\safemath{\delay}{\sctau}					
\safemath{\ddelay}{\scl}						
\safemath{\doppler}{\scnu}					
\safemath{\ddoppler}{\scm}					
\safemath{\freq}{\scf}						
\safemath{\dfreq}{\scn}						
\safemath{\Dtime}{\Delta\time}
\safemath{\Dfreq}{\Delta\freq}
\safemath{\Ddtime}{\Delta\dtime}
\safemath{\Ddfreq}{\Delta\dfreq}
\safemath{\bandwidth}{\scB}
\safemath{\maxdoppler}{\doppler_{0}}			
\safemath{\maxdelay}{\delay_{0}}				
\safemath{\spread}{\Delta_{\CHop}}			
\DeclareMathOperator{\CHop}{\ensuremath{\opH}} 
\safemath{\kernel}{\rndk_{\CHop}}			
\safemath{\kernelp}{\kernel(\time,\time')}	
\safemath{\tvir}{\rndh_{\CHop}}				
\safemath{\tvirp}{\tvir(\time,\delay)}		
\safemath{\tvirc}{\conj{\rndh}_{\CHop}}		
\safemath{\tvtf}{\rndl_{\CHop}}				
\safemath{\tvtfp}{\tvtf(\time,\freq)}			
\safemath{\tvtfc}{\conj{\rndl}_{\CHop}}		
\safemath{\spf}{\rnds_{\CHop}}				
\safemath{\spfp}{\spf(\doppler,\delay)}		
\safemath{\spfc}{\conj{\rnds}_{\CHop}}		
\safemath{\bff}{\rndb_{\CHop}}				
\safemath{\bffp}{\bff(\doppler,\freq)}		

\safemath{\irc}{\scr_{\rndh}}				
\safemath{\tfc}{\scr_{\rndl}}				
\safemath{\spc}{\scr_{\rnds}}				
\safemath{\bfc}{\scr_{\rndb}}				

\safemath{\scaf}{\scc_{\rnds}}				
\safemath{\scafp}{\scaf(\doppler,\delay)}		
\safemath{\ccf}{\scc_{\rndl}}				
\safemath{\ccfp}{\ccf(\Dtime,\Dfreq)}			
\safemath{\cic}{\scc_{\rndh}}				
\safemath{\cicp}{\cic(\Dtime,\delay)}			

\safemath{\mi}{\scI}							
\safemath{\capacity}{\scC}					

\DeclareMathOperator{\Prob}{\opP}		
\safemath{\normal}{\mathcal{N}}			
\safemath{\jpg}{\mathcal{CN}}			
\safemath{\mchain}{\leftrightarrow}		

\safemath{\dB}{\,\mathrm{dB}}
\safemath{\dBm}{\,\mathrm{dBm}}
\safemath{\Hz}{\,\mathrm{Hz}}
\safemath{\kHz}{\,\mathrm{kHz}}
\safemath{\MHz}{\,\mathrm{MHz}}
\safemath{\GHz}{\,\mathrm{GHz}}
\safemath{\s}{\,\mathrm{s}}
\safemath{\ms}{\,\mathrm{ms}}
\safemath{\mus}{\,\mathrm{\text{\textmu}s}}
\safemath{\ns}{\,\mathrm{ns}}
\safemath{\ps}{\,\mathrm{ps}}
\safemath{\meter}{\,\mathrm{m}}
\safemath{\mm}{\,\mathrm{mm}}
\safemath{\cm}{\,\mathrm{cm}}
\safemath{\m}{\,\mathrm{m}}
\safemath{\W}{\,\mathrm{W}}
\safemath{\mW}{\, \mathrm{mW}}
\safemath{\J}{\,\mathrm{J}}
\safemath{\K}{\,\mathrm{K}}
\safemath{\bit}{\,\mathrm{bit}}
\safemath{\nat}{\,\mathrm{nat}}

\safemath{\define}{\triangleq}					

\safemath{\equivalent}{\sim}
\safemath{\distas}{\sim}					
\safemath{\sdiff}{\Delta}				

\safemath{\reals}{\mathbb R}
\safemath{\positivereals}{\reals_{+}}
\safemath{\integers}{\mathbb Z}
\safemath{\posint}{\integers_{+}}
\safemath{\naturals}{\mathbb N}
\safemath{\posnaturals}{\naturals_{+}}
\safemath{\complexset}{\mathbb C}
\safemath{\rationals}{\mathbb Q}

\newcommand*{\fancyrefparlabelprefix}{par}		
\newcommand*{\fancyrefchalabelprefix}{cha}		
\newcommand*{\fancyrefapplabelprefix}{app}		
\newcommand*{\fancyrefthmlabelprefix}{thm}		
\newcommand*{\fancyreflemlabelprefix}{lem}		
\newcommand*{\fancyrefcorlabelprefix}{cor}		
\newcommand*{\fancyrefdeflabelprefix}{def}		

\frefformat{vario}{\fancyrefparlabelprefix}{Part~#1}
\frefformat{vario}{\fancyrefchalabelprefix}{Chapter~#1}
\frefformat{vario}{\fancyrefseclabelprefix}{Section~#1}
\frefformat{vario}{\fancyrefthmlabelprefix}{Theorem~#1}
\frefformat{vario}{\fancyreflemlabelprefix}{Lemma~#1}
\frefformat{vario}{\fancyrefcorlabelprefix}{Corollary~#1}
\frefformat{vario}{\fancyrefdeflabelprefix}{Definition~#1}
\frefformat{vario}{\fancyreffiglabelprefix}{Figure~#1}
\frefformat{vario}{\fancyrefapplabelprefix}{Appendix~#1}
\frefformat{vario}{\fancyrefeqlabelprefix}{(#1)}

\safemath{\iSet}{\setI}
\safemath{\rel}{\bowtie}					
\safemath{\eqrel}{\sim}					
\safemath{\rlord}{\leq}					
\safemath{\slord}{<}						
\safemath{\rpord}{\preceq}				
\safemath{\rrpord}{\succeq}				
\safemath{\spord}{\prec}					
\safemath{\sig}{\sigma}					
\safemath{\metric}{d}					

\safemath{\setfun}{\Phi}					
\safemath{\measure}{\mu}					
\newcommand{\outerm}[1]{#1^{\star}}		
\newcommand{\innerm}[1]{#1_{\star}}		
\safemath{\omeasure}{\outerm{\measure}}		
\safemath{\imeasure}{\innerm{\measure}}		
\safemath{\aecol}{\colS^{\star}_{\measure}} 
\safemath{\emeasure}{\bar{\measure}_{0}}	
\safemath{\rmeasure}{\tilde{\measure}}	
\safemath{\bmeasure}{\measure_{0}}		
\safemath{\glength}{\measure_{\altfun}}	
\safemath{\lebmea}{\lambda}				
\safemath{\blebmea}{\lebmea_{0}}			
\safemath{\sfun}{s}						
\safemath{\absintspace}{\colL^{1}}		
\safemath{\sqintspace}{\colL^{2}}		
\safemath{\abssumspace}{l^{1}}		
\safemath{\sqsumspace}{l^{2}}		
\safemath{\boundspace}{l^{\infty}}	

\safemath{\sfield}{\setF}				
\safemath{\vectors}{\setV}				
\safemath{\vecspace}{(\vectors,\sfield)}	
\safemath{\linspace}{\setV}				
\safemath{\clinspace}{(\linspace,\sfield)} 
\safemath{\nspace}{\setU}				
\safemath{\metspace}{\setM}				
\safemath{\bspace}{\setB}				
\safemath{\ipspace}{\setG}				
\safemath{\hilspace}{\setH}				
\safemath{\blospace}{\setG}				
\safemath{\lop}{\opT}					
\safemath{\altlop}{\opS}					
\safemath{\nullsp}{\nullspace(\lop)}		
\safemath{\lfun}{l}						
\safemath{\altlfun}{g}					
\newcommand{\dual}[1]{#1^{'}}			
\safemath{\dsum}{\oplus}					
\safemath{\funspace}{\colL}				
\renewcommand{\adj}[1]{#1^{\times}}		
\safemath{\adjlop}{\adj{\lop}}			
\safemath{\hadjlop}{\hadj{\lop}}			
\safemath{\tow}{\xrightarrow{w}}			
\safemath{\tows}{\xrightarrow{w^{*}}}		
\safemath{\cparam}{\lambda}
\safemath{\lopl}{\lop_{\cparam}}		
\safemath{\iop}{\opI}					
\safemath{\resolop}{\opR}				
\safemath{\resolvent}{\resolop_{\cparam}(\lop)}	
\safemath{\altresolvent}{\resolop_{\cparam}(\altlop)} 
\safemath{\reset}{\setQ}
\safemath{\spectrum}{\setS}
\safemath{\resolset}{\reset(\lop)}		
\safemath{\lopspec}{\spectrum(\lop)}		
\safemath{\altlopspec}{\spectrum(\altlop)} 
\safemath{\pspec}{\spectrum_{p}(\lop)}	
\safemath{\cspec}{\spectrum_{c}(\lop)}	
\safemath{\rspec}{\spectrum_{r}(\lop)}	
\safemath{\ev}{\cparam}					
\newcommand{\specrad}[1]{r_{#1}}			
\safemath{\lopsrad}{\specrad{\lop}}		
\safemath{\pop}{\opP}					

\safemath{\specfam}{\colE}				
\safemath{\specop}{\opE_{\cparam}}		
\safemath{\altspecop}{\opE_{\mu}}		
\safemath{\vmulti}{\vecone}				
\safemath{\unitaryop}{\opU}				
\safemath{\sval}{\sigma}					
\safemath{\corrcoef}{\rho}				
\safemath{\sangle}{\theta}				

\let\time\undefined
\safemath{\iset}{\setI}				
\safemath{\shift}{\nu}
\safemath{\scale}{\alpha}
\safemath{\time}{t}
\safemath{\specfreq}{\theta}	
\newcommand{\transopgen}[1]{\opT_{#1}} 
\safemath{\transop}{\transopgen{\delay}}
\newcommand{\modopgen}[1]{\opM_{#1}}	
\safemath{\modop}{\modopgen{\shift}}
\newcommand{\dilopgen}[1]{\opD_{#1}}	
\safemath{\dilop}{\dilopgen{\scale}}
\safemath{\fram}{\setF}				
\safemath{\dfram}{\dual{\fram}}		
\safemath{\ufb}{\scB}					
\safemath{\lfb}{\scA}					
\safemath{\sop}{\hadj{\aop}}				
\safemath{\aop}{\opT}			
\safemath{\fop}{\opS}				
\safemath{\daop}{\tilde\opT}			
\safemath{\dsop}{\hadj{\tilde\opT}}				

\safemath{\ifop}{\inv{\fop}}			
\safemath{\rifop}{\fop^{-1/2}}			
\newcommand{\ft}[1]{\widehat{#1}}	
\safemath{\transeq}{\setT}			
\safemath{\nfun}{\Phi}				
\safemath{\funvec}{\vecf}			
\safemath{\altfunvec}{\vecg}

\safemath{\samplespace}{\Omega}
\safemath{\probspace}{(\samplespace,\sfield,\Prob)}	
\safemath{\ccoef}{\rho}			

\safemath{\infstate}{\vecpi}				
\safemath{\typset}{\setA_{\epsilon}^{(N)}}	
\safemath{\expequal}{\doteq}				
\safemath{\code}{C}						
\safemath{\dstringset}{\setD^{\star}}		
\safemath{\cwlength}{l}					
\safemath{\elength}{L}					
\safemath{\extension}{C^{\star}}			
\safemath{\approaches}{\rightarrow}		
\safemath{\evnt}{\setA}					
\safemath{\altevnt}{\setB}					
\safemath{\rv}{\rndx}					
\safemath{\altrv}{\rndy}					
\safemath{\complexrv}{\rndu}					
\safemath{\altcrv}{\rndv}				
\safemath{\rvec}{\rvecx}					
\safemath{\altrvec}{\rvecy}				
\safemath{\crvec}{\rvecu}				
\safemath{\altcrvec}{\rvecv}				
\safemath{\variance}{\sigma^{2}}			
\safemath{\map}{T}						
\safemath{\jacobian}{J}					
\safemath{\wvec}{\rvecw}					
\safemath{\wrv}{\rndw}					
\safemath{\orthmat}{\matQ}				
\safemath{\evmat}{\matLambda}			
\safemath{\identity}{\matidentity}		
\safemath{\innovec}{\vecv}				
\safemath{\convas}{\xrightarrow{\text{a.s.}}}	
\safemath{\convr}{\xrightarrow{\text{r}}}	
\safemath{\convp}{\xrightarrow{\text{P}}}	
\safemath{\convd}{\xrightarrow{\text{D}}}	
\safemath{\ltis}{\opL}				
\safemath{\ir}{h}					
\safemath{\tf}{\MakeUppercase{\ir}}	

{\begin{list}{}%
         {\setlength{\leftmargin}{#1}}%
         \item[]%
}
{\end{list}}

\safemath{\signal}{\scx} 
\safemath{\signalFt}{\widehat{\signal}}	
\safemath{\sigTime}{\signal \! \left( \time \right)}	
\safemath{\orFt}{\widehat{\signal} \! \left(\freq\right)}	
\safemath{\sampFt}{\widehat{\signal}_{\scd} \! \left(\freq\right)}	
\safemath{\fZero}{\freq_{0}}	
\safemath{\fSamp}{\freq_{\scs}}	
\safemath{\specOc}{\setI} 
\safemath{\sampSet}{\setP} 
\safemath{\sampVal}{\signal \! \left(\time_\scn\right)}	
\safemath{\beuDen}{\setD^{-} \! \left(\sampSet\right)}	
\safemath{\samPer}{\scT_\scs}	
\safemath{\hSp}{\setH}	
\safemath{\multSig}{\setB \! \left( \specOc \right)}	
\safemath{\Ltwo}{\sqintspace \! \left(\reals\right)}	
\safemath{\cellNo}{\scL} 
\safemath{\noIn}{s} 
\safemath{\cosetNo}{\scK} 
\safemath{\sampSig}{\scX \! \left(\freq\right)} 
\newcommand{\vandEnt}[1]{\scx_{#1}} 
\safemath{\vandM}{\matV \! \left(\vandEnt{0}, \vandEnt{1}, \ldots, \vandEnt{\cosetNo - 1} \right)} 
\safemath{\sampMat}{\matA} 
\safemath{\suppXhat}{\gamma} 
\safemath{\orFtSupp}{\widehat{\signal}_{\suppXhat} \! \left(\freq\right)} 
\safemath{\maxCard}{\scC} 

\safemath{\dict}{\matD} 
\safemath{\measVec}{\vecy} 
\safemath{\sigVec}{\vecx} 
\safemath{\meas}{\scm} 
\safemath{\sigDim}{\scn} 
\safemath{\spars}{\scs} 
\safemath{\FM}{\matF}	
\safemath{\IM}{\matI}	
\safemath{\fONB}{\matA}	
\safemath{\sONB}{\matB}	
\safemath{\dictCol}{d} 
\safemath{\recSigVec}{\hat{\sigVec}} 
\safemath{\PZ}{\left( \text{P0} \right)} 
\safemath{\BP}{\left( \text{BP} \right)} 
\safemath{\PO}{\left( \text{P1} \right)} 
\safemath{\coher}{\mu} 
\safemath{\supp}{\setS}	
\safemath{\corMeas}{\vecz} 
\safemath{\erVec}{\vece} 
\safemath{\concSig}{\check{\sigVec}} 
\safemath{\fVec}{\vecp} 
\safemath{\sVec}{\vecq} 
\safemath{\conVec}{\vecv} 
\safemath{\sgnl}{\vecs} 
\safemath{\fONBCol}{\veca} 
\safemath{\sONBCol}{\vecb} 
\safemath{\fSupp}{\setP} 
\safemath{\sSupp}{\setQ} 
\safemath{\fDim}{\sigDim_{\fONBCol}} 
\safemath{\sDim}{\sigDim_{\sONBCol}} 
\safemath{\fCoher}{\coher_{\fONBCol}} 
\safemath{\sCoher}{\coher_{\sONBCol}} 
\safemath{\sigSupp}{\setX} 
\safemath{\erSupp}{\setE} 
\safemath{\PZErSup}{\left( \text{P0, } \erSupp \right)} 
\safemath{\falseSigVec}{\tilde{\sigVec}} 
\safemath{\falseErVec}{\tilde{\erVec}} 
\safemath{\sigSpars}{\sigDim_{\sigVec}} 
\safemath{\erSpars}{\sigDim_{\erVec}} 
\safemath{\BPErSup}{\left( \text{BP, } \erSupp \right)} 
\safemath{\gram}{\matG} 

\safemath{\WHT}{\scT}	
\safemath{\WHF}{\scF}	
\safemath{\prot}{\scg}	
\safemath{\vprot}{\vecg}	
\safemath{\proti}{\prot_{\scm, \scn}}	
\safemath{\prott}{\prot \! \left( \time \right) }	
\safemath{\protit}{\proti \! \left(\time\right)}	
\safemath{\vproti}{\vprot_{\scm, \scn}}	
\safemath{\dprot}{\tilde{\scg}}	
\safemath{\dvprot}{\tilde{\vecg}}	
\safemath{\dproti}{\dprot_{\scm, \scn}}	
\safemath{\dprott}{\dprot \! \left( \time \right) }	
\safemath{\dprotit}{\dproti \! \left(\time\right)}	
\safemath{\dvproti}{\dvprot_{\scm, \scn}}	
\safemath{\WO}{\scW}		
\safemath{\WOpam}{\WO^{\left( \WHT, \WHF \right)}_{\scm,\scn}}	
\newcommand{\WOind}[2]{\WO_{#1, #2}} 
\safemath{\WOsind}{\WOind \scm \scn }
\safemath{\zak}{\opZ} 
\safemath{\zakpar}{\zak^{\WHT,\WHF}} 
\safemath{\zaksig}{\zak_{\signal} \! \left( \time, \freq \right)} 
\newcommand{\zakprot}[2]{\zak_{\prot} \! \left( #1, #2 \right)} 
\safemath{\zakprots}{\zakprot{\time}{\freq}}
\safemath{\zakprotis}{\zak_{\prot_{\scm,\scn}} \! \left( \time, \freq \right)} 
\safemath{\tfr}{\scR} 
\safemath{\funmin}{\scm \! \left( \prott; \WHT \right)}	
\safemath{\funmax}{\scM \! \left( \prott; \WHT \right)}	

\safemath{\sclf}{\phi}	
\safemath{\vscf}{\vecphi}	
\newcommand{\scfa}[1]{\sclf \! \left( #1 \right)}	
\safemath{\scft}{\scfa \time}	
\newcommand{\vscfi}[2]{\vscf_{#1,#2}}	
\safemath{\vscfs}{\vscfi \scj \scn}
\safemath{\fscf}{\ft \sclf}	
\safemath{\vfscf}{\ft \vscf}	
\newcommand{\fscfa}[1]{\fscf \! \left( #1 \right)}	
\safemath{\fscff}{\fscfa \freq}	

\safemath{\wav}{\psi}	
\safemath{\vwav}{\vecpsi}	
\newcommand{\wava}[1]{\wav \! \left( #1 \right)}	
\safemath{\wavt}{\wava \time}	
\newcommand{\vwavi}[2]{\vwav_{#1,#2}}	
\safemath{\vwavs}{\vwavi \scj \scn} 
\safemath{\fwav}{\ft \wav}	
\safemath{\vfwav}{\ft \vwav}	
\newcommand{\fwava}[1]{\fwav \! \left( #1 \right)}	
\safemath{\fwavf}{\fwava \freq}	

\newcommand{\renent}[2]{\scH_{#1} \! \left( #2 \right)} 

\safemath{\oprobf}{\hat{p}}	
\safemath{\probf}{p} 

\safemath{\binind}{\setI}			
\safemath{\pool}{\bm{ \mathcal P}}			
\safemath{\poolre}{\setP}			
\safemath{\eptyp}{\setT^{(n)}_\epsilon}	
\safemath{\dist}{\mathbb{P}}			
\safemath{\tdist}{\tilde \dist}			
\newcommand{\distof}[1]{\dist \! \left[ #1 \right]}	
\safemath{\ry}{\setY}			
\safemath{\rz}{\setZ}			
\newcommand{\guess}[2]{\scG_{#1} \! \left( #2 \right)}	
\newcommand{\guessast}[2]{\scG^\ast_{#1} \! \left( #2 \right)}	
\safemath{\trho}{\frac{1}{1+\rho}}
\safemath{\tirho}{\tilde \rho}

\newtheorem{theorem}{Theorem}
\newtheorem{lemma}{Lemma}
\newtheorem{note}[theorem]{Note}
\newtheorem{fact}{Fact}
\newtheorem{corollary}[theorem]{Corollary}

\newcommand{\overbar}[1]{\mkern 1.5mu\overline{\mkern-1.5mu#1\mkern-1.5mu}\mkern 1.5mu}
\ifCLASSINFOpdf
\else
\fi

\usepackage{bbm}
\usepackage{mathrsfs}
\usepackage{cite}

\begin{document}
%
\title{Distributed Storage for Data Security}

\author{\IEEEauthorblockN{Annina Bracher}
\IEEEauthorblockA{ETH Zurich}
\and
\IEEEauthorblockN{Eran Hof}
\IEEEauthorblockA{Samsung Israel Research and Development Center}
\and
\IEEEauthorblockN{Amos Lapidoth}
\IEEEauthorblockA{ETH Zurich}}


%


\maketitle

\begin{abstract}
We study the secrecy of a distributed storage system for passwords. The encoder, Alice, observes a length-$n$ password and describes it using two hints, which she then stores in different locations. The legitimate receiver, Bob, observes both hints. The eavesdropper, Eve, sees only one of the hints; Alice cannot control which. We characterize the largest normalized (by $n$) exponent that we can guarantee for the number of guesses it takes Eve to guess the password subject to the constraint that either the number of guesses it takes Bob to guess the password or the size of the list that Bob must form to guarantee that it contain the password approach $1$ as $n$ tends to infinity. 
\end{abstract}


%
\IEEEpeerreviewmaketitle

\section{Introduction}

Suppose that some sensitive information $X$ (e.g.\ password) is drawn from a finite set $\setX$ according to some PMF $P_X$. A (stochastic) encoder, Alice, maps (possibly using randomization) $X$ to two hints $M_1$ and $M_2$, which she then stores in different locations. The hints are intended for a legitimate receiver, Bob, who knows where they are stored and has access to both. An eavesdropper, Eve, sees one of the hints but not both; we do not know which. Given some notion of ambiguity, we would ideally like Bob's ambiguity about $X$ to be small and Eve's large.

Which hint is revealed to Eve is a subtle question. We adopt a conservative approach and assume that, after observing $X$, an adversarial ''genie'' reveals to Eve the hint that minimizes her ambiguity. Not allowing the genie to observe $X$ would lead to a weaker form of secrecy (an example is given in \cite{bracherlapidoth14}).

There are several ways to define ambiguity. For example, we could require that Bob be able to reconstruct $X$ whenever $X$ is ''typical'' and that the conditional entropy of $X$ given Eve's observation be large. For some scenarios, such an approach might be inadequate. Firstly, this approach may not properly address Bob's needs when $X$ is not typical. For example, if Bob must guess $X$, this approach does not guarantee that the expected number of guesses be small: It only guarantees that the probability of success after one guess be large. It does not indicate the number of guesses that Bob might need when $X$ is atypical. Secondly, conditional entropy need not be an adequate measure of Eve's ambiguity: if $X$ is some password that Eve wishes to uncover, then we may care more about the number of guesses that Eve needs than about the conditional entropy \cite{arikanmerhav99}.

In this paper, we assume that Eve wants to guess $X$ with the least number of guesses of the form "Is X = x?". We quantify Eve's ambiguity about $X$ by the expected number of guesses that she needs to uncover $X$. In this sense, Eve faces an instance of Arikan's guessing problem \cite{arikan96}. For each possible observation $z$ in some finite set $\setZ$, Eve chooses a guessing function $\guess {} { \cdot \left| z \right.}$ from $\setX$ onto the set $\set{1, \ldots, \card \setX}$, which determines the guessing order: if Eve observes $z$, then the question ''Is $X = x$?'' will be her $\guess {}{x|z}$-th question. Eve's expected number of guesses is $\Ex {}{\guess {}{X|Z}}$. This expectation is minimized if for each $z \in \setZ$ the guessing function $\guess {} { \cdot \left| z \right.}$ orders the possible realizations of $X$ in decreasing order of their posterior probabilities given $Z = z$.

As to Bob, we will consider two different criteria: In the ''guessing version'' of the problem the criterion is the expected number of guesses it takes Bob to guess $X$, and in the ''list version'' the criterion is the first moment of the size of the list that Bob must form to guarantee that it contain $X$. We shall see that the two criteria lead to similar results.

The former criterion is natural when Bob can check whether a guess is correct: If $X$ is some password, then Bob can stop guessing as soon as he has gained access to the account that is secured by $X$.

The latter criterion is appropriate if Bob does not know whether a guess is correct. For example, if $X$ is a task that Bob must perform, then the only way for Bob to make sure that he performs $X$ is to perform all the tasks in a list comprising the tasks that have positive posterior probabilities given his observation. In this scenario, a good measure of Bob's ambiguity about $X$ is the expected number of tasks that he must perform, and this will be small whenever Alice is a good task-encoder for Bob \cite{buntelapidoth14}. To describe the list of tasks that Bob must perform more explicitly, let us denote by
\be 
\distof{M_1 = m_1, M_2 = m_2 \left| X = x \right.}, \, m_1 \in \setM_1, \, m_2 \in \setM_2
\een
the probability that Alice produces the pair of hints $\left( M_1, M_2 \right) = \left( m_1,m_2 \right)$ upon observing that $X = x$. It is 0-1 valued if Alice does not use randomization. Upon observing that $\left( M_1, M_2 \right) = \left( m_1,m_2 \right)$, Bob produces the list $\setL_{m_1,m_2}$ of all the tasks $x \in \setX$ whose posterior probability $\distof{ X = x | M_1 = m_1, M_2 = m_2 }$ is positive. Our notion of Bob's ambiguity about $X$ is $\bigEx {}{\card {\setL_{M_1,M_2}}}$.

The guessing and the list-size criterion for Bob lead to similar results in the following sense: Clearly, every guessing function $\guess {}{\cdot | M_1,M_2}$ for $X$ that maps the elements of $\setX$ that have zero posterior probability to larger values than those that have a positive posterior probability satisfies $\bigEx {}{\guess {}{X | M_1,M_2}} \leq \bigEx {}{\card {\setL_{M_1,M_2}}}$. Conversely, one can prove that every pair of ambiguities for Bob and Eve that is achievable in the ''guessing version'' is, up to polylogarithmic factors of $\card \setX$, also achievable in the ''list version'' provided that we increase $\setM_1$ or $\setM_2$ by a logarithmic factor of $\card \setX$ (see Section~\ref{sec:guessListClose} ahead). These polylogarithmic factors wash out in the asymptotic regime where the sensitive information is an $n$-tuple and $n$ tends to infinity.

With no extra effort we can generalize the model and replace expectations with $\rho$-th moments. This we do to better bring out the role of R\'{e}nyi entropy. For an arbitrary $\rho > 0$, we thus study the $\rho$-th (instead of the first) moment of the list-size and of the number of guesses. Moreover, we shall allow some side-information $Y$ that is available to all parties. We shall thus assume that $\left( X,Y \right)$ take value in the finite set $\setX \times \setY$ according to $P_{X,Y}$.

\section{Problem Statement}

We consider two problems, which we call the ''guessing version'' and the ''list version''. They differ in the definition of Bob's ambiguity. In both versions a pair $\left( X,Y \right)$ is drawn from the finite set $\setX \times \setY$ according to the PMF $P_{X,Y}$, and $\rho > 0$ is fixed. Upon observing $\left( X,Y \right) = \left( x,y \right)$, Alice draws the hints $M_1$ and $M_2$ from the finite set $\setM_1 \times \setM_2$ according to some PMF
\be 
\distof{M_1 = m_1, M_2 = m_2 \left| X = x, Y = y \right.}. \label{eq:aliceEncPMF}
\ee
In the ''guessing version'' Bob's ambiguity about $X$ is
\be 
\mathscr A^{(\text g)}_{\text{B}} \! \left( P_{X,Y} \right) = \min_{G} \bigEx {}{\guess {}{X | Y,M_1,M_2}^\rho}. \label{eq:bobAmbiguityGuessing}
\ee
In the ''list version'' Bob's ambiguity about $X$ is
\be 
\mathscr A^{(\text l)}_{\text{B}} \! \left( P_{X,Y} \right) = \bigEx {}{| \setL^Y_{M_1,M_2} |^\rho}, \label{eq:bobAmbiguityList}
\ee
where for all $y \in \setY$ and $\left( m_1,m_2 \right) \in \setM_1 \times \setM_2$
\be 
\setL^y_{m_1,m_2} =  \set{x \colon \distof{ X = x \left| Y = y, M_1 = m_1, M_2 = m_2 \right.} > 0}
\een
is the list of all the tasks whose posterior probability
\ba
&\distof{ X = x \left| Y = y, M_1 = m_1, M_2 = m_2 \right.} \nonumber \\
&\,= \! \frac{P_{X,Y} \! \left( x,y \right) \distof{M_1 \! = \! m_1, M_2 \! = \! m_2 \left| X \! = \! x, Y \! = \! y \right.}}{\sum_{\tilde x} P_{X,Y} \! \left( \tilde x,y \right) \distof{M_1 \! = \! m_1, M_2 \! = \! m_2 \left| X \! = \! \tilde x, Y \! = \! y \right.}}
\ea
is positive. In both versions Eve's ambiguity about $X$ is
\be 
\mathscr A_{\text{E}} \! \left( P_{X,Y} \right) \! = \! \min_{G_1,G_2} \! \bigEx {}{\guess 1 { X \left| Y, M_1 \right. }^\rho \! \wedge \guess 2 { X \left| Y, M_2 \right. }^\rho}, \label{eq:distEncSecrecyMeasure}
\ee
where $\alpha \wedge \beta$ denotes the minimum of $\alpha$ and $\beta$.

Optimizing over Alice's mapping, i.e., the choice of the conditional PMF in \eqref{eq:aliceEncPMF}, we wish to characterize the largest ambiguity that we can guarantee that Eve will have subject to a given upper bound on the ambiguity that Bob may have.

Of special interest to us is the asymptotic regime where $\left( X,Y \right)$ is an $n$-tuple (not necessarily drawn IID), and where $$\setM_1 = \left\{ 1, \ldots, 2^{n R_1} \right\}, \, \setM_2 = \left\{ 1, \ldots, 2^{n R_2} \right\},$$ where $\left( R_1,R_2 \right)$ is a nonnegative pair corresponding to the rate. For both versions of the problem, we shall characterize the largest exponential growth that we can guarantee for Eve's ambiguity subject to the constraint that Bob's ambiguity tend to one. This asymptote turns out not to depend on the version of the problem, and in the asymptotic analysis $\mathscr A_{\text{B}}$ can stand for either $\mathscr A_{\text{B}}^{(\text g)}$ or $\mathscr A_{\text{B}}^{(\text l)}$.

To phrase this mathematically, let us introduce the stochastic process $\set{\left(X_i,Y_i\right)}_{i \in \naturals}$ with finite alphabet $\setX \times \setY$. We denote by $P_{X^n,Y^n}$ the PMF of $\left( X^n,Y^n \right)$. For a nonnegative rate-pair $\left( R_1,R_2 \right)$, we call $E_{\text{E}}$ an \emph{achievable ambiguity-exponent} if there is a sequence of stochastic encoders such that Bob's ambiguity (which is always at least $1$) satisfies
\ba 
&\lim_{n \rightarrow \infty} \mathscr A_{\text{B}} \! \left( P_{X^n,Y^n} \right) = 1, \label{eq:bobAmbTo1}
\ea
and such that Eve's ambiguity satisfies
\ba 
\liminf_{n \rightarrow \infty} \frac{\log \! \left( \mathscr A_{\text{E}} \! \left( P_{X^n,Y^n} \right) \right)}{n} &\geq E_{\text{E}}. \label{eq:eveGuessBobAmbTo1}
\ea
We shall characterize the supremum $\overbar{E_{\text{E}}}$ of all achievable ambiguity-exponents, which we call \emph{privacy-exponent}. If \eqref{eq:bobAmbTo1} cannot be satisfied, then the set of achievable ambiguity-exponents is empty, and we say that the privacy-exponent is negative infinity.

\section{Main Results}

To describe our results, we shall need a conditional version of R\'enyi entropy (originally proposed by Ariomoto \cite{arimoto77} and also studied in \cite{buntelapidoth14})
\be 
\renent {\alpha}{ X | Y } = \frac{\alpha}{1 - \alpha} \log \! \sum_{y \in \setY} \Bigl( \sum_{x \in \setX} P_{X,Y} \! \left( x,y \right)^{\alpha} \Bigr)^{1/\alpha},
\ee
where $\alpha \in \left[ 0,\infty \right]$ is the order and where the cases where $\alpha$ is $0$, $1$, or $\infty$ are treated by a limiting argument. In addition, we shall need the notion of conditional R\'enyi entropy-rate: Let $\set{\left(X_i,Y_i\right)}_{i \in \naturals}$ be a discrete-time stochastic process with finite alphabet $\setX \times \setY$. Whenever the limit as $n$ tends to infinity of $\renent {\alpha}{ X^n | Y^n } / n$ exists, we denote it by $\renent {\alpha}{ \rndvecX | \rndvecY }$ and call it conditional R\'enyi entropy-rate. In this paper, $\alpha = 1 / (1+\rho)$ takes value in the set $\left( 0,1 \right)$. To simplify notation, we henceforth write $\tirho$ for $1 / (1+\rho)$ and $\alpha \vee \beta$ for the maximum of $\alpha$ and $\beta$.

\subsection{Finite Blocklength Results}

In the next two theorems $c_s$ is related to how much can be gleaned about $X$ from $( M_1,M_2)$ but not from one hint alone; $c_1$ is related to how much can be gleaned from $M_1$; and $c_2$ is related to how much can be gleaned from $M_2$. More precisely, we shall see in Section~\ref{sec:proofs} ahead that Alice first maps $( X,Y )$ to the triple $( V_s, V_1, V_2 )$, which takes value in a set $\setV_s \times \setV_1 \times \setV_2$ of size $c_s c_1 c_2$. Independently of $( X,Y )$ she then draws a (one-time-pad like) random variable $U$ uniformly over $\setV_s$ and maps $( U, V_s )$ to a variable $\tilde V_s$ choosing the (XOR like) mapping so that $V_s$ can be recovered from $( \tilde V_s, U )$ while $\tilde V_s$ alone is independent of $( X,Y )$. The hints are $M_1 \! = \! ( \tilde V_s, V_1 )$ and $M_2 \! = \! ( U, V_2 )$. Alice does not use randomization if $c_s = 1$.

\begin{theorem}[Finite Blocklength Guessing Version]\label{th:distStorGuess}
For every triple $\left( c_s, c_1, c_2 \right) \in \naturals^3$ satisfying
\ba
&c_s \! \leq \! \card {\setM_1} \! \wedge \! \card {\setM_2},\, c_1 \! \leq \! \left\lfloor \card {\setM_1} \! / c_s \right\rfloor, \, c_2 \! \leq \! \left\lfloor \card {\setM_2} \! / c_s \right\rfloor, \label{bl:condCsC1C2Guess}
\ea
there is a choice of the conditional PMF in \eqref{eq:aliceEncPMF} for which Bob's ambiguity about $X$ is upper-bounded by
\ba 
\mathscr A_{\text{B}}^{(\text g)} \! \left( P_{X,Y} \right) < 1 + 2^{\rho \left( \renent {\tirho} {X \left| Y \right.} - \log \! \left( c_s c_1 c_2 \right) + 1 \right)}, \label{eq:BobMomDistStorDirGuess}
\ea
and Eve's ambiguity about $X$ is lower-bounded by
\ba 
&\mathscr A_{\text{E}} \! \left( P_{X,Y} \right) \geq \left( 1 + \ln \! \card \setX \right)^{-\rho} 2^{\rho \left( \renent {\tirho} { X \left| Y \right.} - \log \! \left( c_1 + c_2 \right) \right)}. \label{eq:EveMomDistStorDirGuess}
\ea
Conversely, for every conditional PMF, Bob's ambiguity is lower-bounded by
\ba
\!\!\!\! \mathscr A_{\text{B}}^{(\text g)} \!\! \left( P_{X,Y} \right) \! \geq \! \left( 1 \! + \! \ln \! \card \setX \right)^{\!-\!\rho} \! 2^{\rho \left( \renent {\tirho} {X \left| Y \right.} - \log \! \card {\setM_1} \card {\setM_2} \right)} \!\! \vee \! 1, \label{eq:BobMomDistStorConvGuess}
\ea
and Eve's ambiguity is upper-bounded by
\ba 
&\! \mathscr A_{\text{E}} \! \left( P_{X,Y} \right) \! \leq \! \left( \card {\setM_1}^\rho \!\! \wedge \! \card {\setM_2}^\rho \right) \! \mathscr A^{(\text g)}_{\text B} \! \left( P_{X,Y} \right) \! \wedge \! 2^{\rho \renent {\tirho} {X \left| Y \right.}}. \label{eq:EveMomDistStorConvGuess}
\ea
\end{theorem}

\begin{theorem}[Finite Blocklength List Version]\label{th:distStor} \hspace{-2mm} If $| \! \setM_1 \! | | \! \setM_2 \! | \!\! > \log \! \card \setX + 2$, then for every triple $\left( c_s, c_1, c_2 \right) \in \naturals^3$ satisfying
\begin{subequations}\label{bl:condCsC1C2}
\ba
&c_s \! \leq \! \card {\setM_1} \! \wedge \! \card {\setM_2},\, c_1 \! \leq \! \left\lfloor \card {\setM_1} \! / c_s \right\rfloor, \, c_2 \! \leq \! \left\lfloor \card {\setM_2} \! / c_s \right\rfloor, \\
&c_s c_1 c_2 > \log \! \card \setX + 2, \label{eq:condCsC1C2List}
\ea
\end{subequations}
there is a choice of the conditional PMF in \eqref{eq:aliceEncPMF} for which Bob's ambiguity about $X$ is upper-bounded by
\ba 
\!\! \mathscr A_{\text{B}}^{(\text l)} \! \left( P_{X,Y} \right) < 1 + 2^{\rho \left( \renent {\tirho} {X \left| Y \right.} - \log \! \left( c_s c_1 c_2 - \log \! \card \setX - 2 \right) + 2 \right)}, \label{eq:BobMomDistStorDir}
\ea
and Eve's ambiguity about $X$ is lower-bounded by
\ba 
&\mathscr A_{\text{E}} \! \left( P_{X,Y} \right) \geq \left( 1 + \ln \! \card \setX \right)^{-\rho} 2^{\rho ( \renent {\tirho} { X \left| Y \right.} - \log \! \left( c_1 + c_2 \right) )}. \label{eq:EveMomDistStorDir}
\ea 
Conversely, for every conditional PMF, Bob's ambiguity is lower-bounded by
\ba
\mathscr A_{\text{B}}^{(\text l)} \! \left( P_{X,Y} \right) \geq 2^{\rho \left( \renent {\tirho} {X \left| Y \right.} - \log \! \card {\setM_1} \card {\setM_2} \right)} \vee 1, \label{eq:BobMomDistStorConv}
\ea
and Eve's ambiguity is upper-bounded by
\ba 
&\mathscr A_{\text{E}} \! \left( P_{X,Y} \right) \! \leq \! \left( \card {\setM_1}^\rho \!\! \wedge \! \card {\setM_2}^\rho \right) \! \mathscr A^{(\text l)}_{\text B} \! \left( P_{X,Y} \right) \! \wedge \! 2^{\rho \renent {\tirho} {X \left| Y \right.}}. \label{eq:EveMomDistStorConv}
\ea
\end{theorem}

We sketch a proof of Theorems~\ref{th:distStorGuess} and \ref{th:distStor} in Section~\ref{sec:proofs} ahead. Here, we discuss an important implication of the theorems:

\begin{note}\label{no:impPolyLog}
We can choose the conditional PMF in \eqref{eq:aliceEncPMF} so that, neglecting polylogarithmic factors of $\card \setX$, Bob's ambiguity satisfies any prespecified upper bound $\mathscr U_{\text{B}}$ no smaller than $2^{\rho ( \renent {\tirho} {X \left| Y \right.} - \log \! \card {\setM_1} \card {\setM_2} )} \vee 1$, while Eve's ambiguity is guaranteed to be $( \card {\setM_1}^\rho \!\! \wedge \! \card {\setM_2}^\rho ) \mathscr U_B \! \wedge \! 2^{\rho \renent {\tirho} {X \left| Y \right.}}$.
\end{note}

To show that Note~\ref{no:impPolyLog} holds, we next argue that the bounds in Theorems~\ref{th:distStorGuess} and \ref{th:distStor} are tight in the sense that with a judicious choice of $\left( c_s,c_1,c_2 \right)$ the achievability results (namely \eqref{eq:BobMomDistStorDirGuess}--\eqref{eq:EveMomDistStorDirGuess} in the ''guessing version'' and \eqref{eq:BobMomDistStorDir}--\eqref{eq:EveMomDistStorDir} in the ''list version'') match the corresponding converse results (namely \eqref{eq:BobMomDistStorConvGuess}--\eqref{eq:EveMomDistStorConvGuess} in the ''guessing version'' and \eqref{eq:BobMomDistStorConv}--\eqref{eq:EveMomDistStorConv} in the ''list version'') up to polylogarithmic factors of $\card \setX$. By possibly relabeling the hints, we can assume w.l.g.\ that $\card {\setM_2} \leq \card {\setM_1}$. If $\card {\setM_2}$ exceeds $2^{\renent {\tirho} { X \left| Y \right.}}$ we can choose $( c_s, c_1, c_2 ) = ( \card {\setM_2}, 1, 1)$. Neglecting polylogarithmic factors of $\card \setX$, this choice guarantees that Bob's ambiguity be close to one, while Eve's ambiguity is $2^{\rho \renent {\tirho} {X \left| Y \right.}}$. Suppose now that $\card {\setM_2}$ does not exceed $2^{\renent {\tirho} { X \left| Y \right.}}$. In this case we can choose $\left( c_s,c_1,c_2 \right)$ so that $c_1 \geq c_2$ and, neglecting logarithmic factors of $\card \setX$, so that $c_s c_2 = \card {\setM_2}$ while $c_s c_1 c_2$ assumes any given integer value between $\card {\setM_2}$ and $( \card {\setM_1} \card {\setM_2} ) \wedge 2^{\renent {\tirho} { X \left| Y \right.}}$. This indeed proves the claim: neglecting polylogarithmic factors of $\card \setX$, we can guarantee that Bob's ambiguity satisfy any given upper bound no smaller than the RHS of \eqref{eq:BobMomDistStorConvGuess} or \eqref{eq:BobMomDistStorConv}, while Eve's ambiguity satisfies \eqref{eq:EveMomDistStorConvGuess} or \eqref{eq:EveMomDistStorConv} with equality.

\subsection{Asymptotic Results}

Consider now the asymptotic regime where $\left( X,Y \right)$ is an $n$-tuple. In this case the results are the same for both versions of the problem, and we thus refer to both $\mathscr A^{(\text g)}_{\text{B}}$ and $\mathscr A^{(\text l)}_{\text{B}}$ by $\mathscr A_{\text{B}}$. With a judicious choice of $\left( c_s,c_1,c_2 \right)$ one can show that Theorems~\ref{th:distStorGuess} and \ref{th:distStor} imply the following asymptotic result:

\begin{corollary}\label{co:asympDistStor}
Let $\set{\left(X_i,Y_i\right)}_{i \in \naturals}$ be a discrete-time stochastic process with finite alphabet $\setX \times \setY$, and suppose its conditional R\'enyi entropy-rate $\renent {\tirho}{ \rndvecX | \rndvecY }$ is well-defined. Given any positive rate-pair $\left( R_1, R_2 \right)$, the privacy-exponent is
\ba 
\! \overbar{E_{\text{E}}} \! = \! \begin{cases} \rho \! \left( R_1 \! \wedge \! R_2 \! \wedge \! \renent {\tirho}{ \rndvecX | \rndvecY } \right) \!, &\!\!\!\! R_1 \! + \! R_2 \! > \! \renent {\tirho}{ \rndvecX | \rndvecY } \\ - \infty, &\!\!\!\! R_1 \! + \! R_2 \! < \! \renent {\tirho}{ \rndvecX | \rndvecY }. \end{cases} \label{eq:secStrongConst}
\ea
\end{corollary}

In the full version of this paper \cite{bracherlapidoth14}, we generalize Corollary~\ref{co:impGuess} to a scenario where Bob's ambiguity may grow exponentially with a given normalized (by $n$) exponent $E_{\text{B}}$.

\section{Lists and Guesses}

The results for the ''guessing version'' and the ''list version'' are remarkably similar. To understand why, we relate task-encoders to guessing functions. We show that a good guessing function induces a good task-encoder and vice versa:

\begin{theorem}\label{th:relGuessEnc}
Let $\left( X,Y \right)$ be drawn from the finite set $\setX \times \setY$ according to the PMF $P_{X,Y}$. Using the side-information $Y$, a stochastic task-encoder describes the task $X$ by a chance variable $M$, which it draws from a finite set $\setM$ according to some conditional PMF
\be 
\distof{M = m | X = x, Y = y}, \, m \in \setM, \, x \in \setX, \, y \in \setY. \label{eq:condPMFRelGuessEnc}
\ee
For any PMF \eqref{eq:condPMFRelGuessEnc} define for all $m \in \setM$ and $y \in \setY$ the lists
\be 
\setL^y_m = \set{ x \in \setX \colon \distof {X = x | Y = y, M = m} > 0 }. \label{eq:relGuessEncLists}
\ee
\begin{enumerate}
\item For every conditional PMF \eqref{eq:condPMFRelGuessEnc} the lists $\set{\setL^y_m}$ induce a guessing function $\guess {}{\cdot | Y }$ for $X$ such that
\be 
\Ex {}{\guess {} {X | Y }^\rho} \leq \card \setM^\rho \bigEx {}{ | \setL^Y_M |^\rho}. \label{eq:listToGuess}
\ee
\item Every guessing function $\guess {}{\cdot | Y }$ for $X$ and every positive integer $v \leq \card \setX$ satisfying
\be 
\card \setM \geq v \left( \left\lfloor \log \! \left\lceil \card \setX / v \right\rceil \right\rfloor + 1 \right) \label{eq:relCardMandV}
\ee
induce a $0$-$1$ valued conditional PMF \eqref{eq:condPMFRelGuessEnc}---i.e., a deterministic task-encoder---whose lists $\set{\setL^y_m}$ satisfy
\be 
\bigEx {}{| \setL^Y_M |^\rho} \leq \Ex {}{\left\lceil \guess {} {X | Y } / v \right\rceil^\rho}. \label{eq:guessToList}
\ee
\end{enumerate}
\end{theorem}

To prove Theorem~\ref{th:relGuessEnc}, we need the following fact:

\begin{fact}\label{fa:1}
Fix a positive integer $u$, and let $h \! \left( \cdot \right)$ map every $k \in \set{1, \ldots, u}$ to $\left\lfloor \log \! k \right\rfloor$. Then,
\be\label{eq:fa1}
\bigr|\bigr\{\tilde k \! \in \! \set{1, \ldots, u} \colon h ( \tilde k ) = h \! \left( k \right) \! \bigl\}\bigl| \leq k, \, k \in \set{1, \ldots, u}.
\ee
\end{fact}

\begin{proof}
If $k, \tilde k \in \set{1, \ldots, u}$ are such that $h (\tilde k) = h (k)$, then $2^{\left\lfloor \log \! k \right\rfloor} \leq k < 2^{\left\lfloor \log \! k \right\rfloor + 1}$. Hence, \eqref{eq:fa1} holds.
\end{proof}

%

\begin{proof}[Proof of Theorem~\ref{th:relGuessEnc}] As to the first part, suppose we are given a conditional PMF \eqref{eq:condPMFRelGuessEnc} with corresponding lists $\set{ \setL^y_m }$ as in \eqref{eq:relGuessEncLists}. For each $y \in \setY$, order the lists $\set{\setL^y_m}_{m \in \setM}$ in increasing order of their cardinalities, and order the elements in each list in some arbitrary way. Now consider the guessing order where we first guess the elements of the first (and smallest) list in their respective order followed by those elements in the second list that have not yet been guessed (i.e., that are not contained in the first list) and we continue until concluding by guessing those elements of the last (and longest) list that have not been previously guessed. Let $\guess {}{\cdot | Y}$ be the corresponding guessing function, and observe that
\bas 
\Ex {}{\guess {}{X|Y}^\rho} &= \sum_{x,y} P_{X,Y} \! \left( x,y \right) \card {\set{\tilde x \colon \guess {}{\tilde x|y} \leq \guess {}{x|y}}}^\rho \\
&\stackrel{(a)}\leq \sum_{x,y} P_{X,Y} \! \left( x,y \right) \card \setM^\rho \min_{m \colon x \in \setL^y_m} \card {\setL^y_m}^\rho \\
&\leq \card \setM^\rho \bigEx {}{| \setL^Y_M |^\rho},
\eas
where $(a)$ holds because for all $x, \tilde x \in \setX$ and $y \in \setY$ a necessary condition for $\guess {}{\tilde x|y} \leq \guess {}{x|y}$ is that $\tilde x \in \setL^y_{\tilde m}$ for some $\tilde m \in \setM$ satisfying $\card{\setL^y_{\tilde m}} \leq \min_{m \colon x \in \setL^y_m} \card {\setL^y_m}$, and the number of lists whose size does not exceed $\min_{m \colon x \in \setL^y_m} \card {\setL^y_m}$ is at most $\card \setM$.

As to the second part, suppose we are given a guessing function $\guess {}{\cdot | Y }$ for $X$ and a positive integer $v \leq \card \setX$ that satisfies \eqref{eq:relCardMandV}. Let $\setZ = \left\{ 0, \ldots, v - 1 \right\}$ and $\setS = \left\{ 0, \ldots, \left\lfloor \log \! \left\lceil \card \setX / v \right\rceil \right\rfloor \right\}$. From \eqref{eq:relCardMandV} it follows that $\card \setM \geq \card \setZ \card \setS$. It thus suffices to prove the existence of a task-encoder that uses only $\card \setZ \card \setS$ possible descriptions, and we thus assume w.l.g.\ that $\setM = \setZ \times \setS$. That is, using the side-information $y$ the task-encoder (deterministically) describes $x$ by $m = (z,s)$. The encoding involves two steps:

\textbf{Step~1:} In Step~1 the encoder first computes $Z \in \setZ$ as the remainder of the Euclidean division of $\guess {}{X|Y} - 1$ by $\card \setZ$. It then constructs from $\guess {}{\cdot | Y}$ a guessing function $\guess {}{\cdot | Y,Z}$ for $X$ as follows. Given $Y = y$ and $Z = z$, the task $X$ must be in the set $\setX_{y,z} \triangleq \set{x \in \setX \colon ( \guess {} { x | y } - 1 ) \equiv z \, \bmod \card \setZ}$. The encoder constructs the guessing function $\guess {}{ \cdot | y,z }$ so that---in the corresponding guessing order---we first guess the elements of $\setX_{y,z}$ in increasing order of $\guess {}{ x | y }$. For $l \in \set{1, \ldots, \card {\setX_{y,z}}}$ our $l$-th guess $x_l$ is thus the element of $\setX_{y,z}$ for which $\guess {}{x_l|y} = z + 1 + ( l - 1 ) \card \setZ$. 
Once we have guessed all the elements of $\setX_{y,z}$ we guess the remaining elements of $\setX$ in some arbitrary order. This order is immaterial because $X$ is guaranteed to be in the set $\setX_{y,z}$. 
Since $z + 1 \in \set{ 1, \ldots, \card \setZ}$ we find that $\guess {}{ x | y,z } = \left\lfloor \guess {}{ x | y } / \card \setZ \right\rfloor$ whenever $x \in \setX_{y,z}$. But $X$ is guaranteed to be in the set $\setX_{y,z}$. Hence, the guessing function $\guess {}{\cdot | Y,Z}$ for $X$ satisfies
\be\label{eq:guessToListZ}
\guess {} {X | Y,Z } = \left\lceil \guess {} { X | Y } / \card \setZ \right\rceil.
\ee

\textbf{Step~2:} In Step~2 the encoder first computes $S = \left\lfloor \log \! \guess {}{X | Y, Z} \right\rfloor \in \setZ$, and then describes the task $X$ by $M \triangleq \left( Z,S \right)$. Given $Y = y$, $Z = z$, and $S = s$, the task $X$ must be in the set $\setX^{y,z}_{s} \triangleq \set{x \in \setX \colon \left\lfloor \log \! \guess {}{x | y, z} \right\rfloor = s}$. Fact~\ref{fa:1} and the fact that the guessing function $\guess {}{\cdot | y,z }$ is a bijection imply that $\card {\setX^{y,z}_s} \leq \guess {}{ x | y, z}$ for all $x \in \setX^{y,z}_s$. Since $X$ is guaranteed to be in the set $\setX^{y,z}_s$ we have $| \setX^{Y,Z}_S | \leq \guess {}{ X | Y, Z}$. From $M = (Z,S)$ and \eqref{eq:relGuessEncLists} we obtain that the list $\setL^Y_M$ is contained in the set $\setX^{Y,Z}_S$ and thus satisfies
\be\label{eq:guessToListS}
| \setL^Y_M | \leq \guess {} { X | Y,Z }.
\ee

Recalling that $\card \setZ = v$ we conclude from \eqref{eq:guessToListZ} and \eqref{eq:guessToListS}
\be 
\bigEx {}{| \setL^Y_M |^\rho } \leq \Ex {}{\guess {} { X | Y,Z }} = \Ex {}{\left\lceil \guess {} {X | Y } / v \right\rceil^\rho}.
\ee
Since $Z$ and $S$ are deterministic given $\left( X,Y \right)$ the conditional PMF \eqref{eq:condPMFRelGuessEnc} associated with $M = (Z,S)$ is $0$-$1$ valued.
\end{proof}

The choice of $v$ as $\left\lfloor \card \setM / \left( \left\lfloor \log \! \card \setX \right\rfloor + 1 \right) \right\rfloor$ and \cite[Equation~(26)]{buntelapidoth14}, i.e., $\left\lceil \xi \right\rceil^\rho < 1 + 2^{\rho} \xi^\rho, \xi \geq 0$, imply our next result:

\begin{corollary}\label{co:guessToBestList}
Every guessing function $\guess {}{\cdot | Y}$ for $X$ induces a deterministic task-encoder corresponding to a $0$-$1$ valued conditional PMF \eqref{eq:condPMFRelGuessEnc} that satisfies
\be 
\bigEx {}{| \setL^Y_M |^\rho} \! \leq \! 1 \! + \! 2^{\rho} \Ex {}{\guess {}{X | Y}^\rho} \left( \! \frac{\card \setM}{\log \! \card \setX + 1} \! - \! 1 \! \right)^{-\rho}. \label{eq:guessToBestList}
\ee
\end{corollary}

Combined with Arikan's bounds \cite[Theorem~1 and Proposition~4]{arikan96} on $\Ex {}{\guess {}{X | Y}^\rho}$, Equations~\eqref{eq:listToGuess} and \eqref{eq:guessToBestList} provide an upper and a lower bound on the smallest $\Ex {}{| \setL^Y_M |^\rho}$ that is achievable for a given $\card \setM$. These bounds are weaker than \cite[Theorem~1.1 and Theorem~6.1]{buntelapidoth14} in the finite blocklength regime but tight enough to prove the asymptotic results \cite[Theorem~1.2 and Theorem~6.2]{buntelapidoth14}.

Another interesting corollary to Theorem~\ref{th:relGuessEnc} results from the choice of $v$ as $1$ (see Section~\ref{sec:guessListClose} for an implication):

\begin{corollary}\label{co:guessToList}
For $\card \setM = \left\lfloor \log \! \card \setX \right\rfloor + 1$ every guessing function $\guess {}{ \cdot | Y}$ induces a deterministic task-encoder for which
\be 
\bigEx {}{| \setL^Y_M |^\rho} \leq \bigEx {}{\guess {} {X | Y }^\rho}. \label{eq:guessToListV1}
\ee
\end{corollary}

\section{On the Proof of Theorems~\ref{th:distStorGuess} and \ref{th:distStor}}\label{sec:proofs}

To prove Theorems~\ref{th:distStorGuess} and \ref{th:distStor}, we must quantify how additional side-information $Z$ helps guessing. We show that if $Z$ takes value in a finite set $\setZ$, then it can reduce the $\rho$-th moment of the number of guesses by at most a factor of $\card \setZ^{-\rho}$.

\begin{lemma}\label{le:ImproveGuess}
Let $\left( X,Y,Z \right)$ be drawn from the finite set $\setX \times \setY \times \setZ$ according to the PMF $P_{X,Y,Z}$. Then,
\be
\Ex {}{\guessast {}{X|Y,Z}^\rho} \geq \bigEx {}{\left\lceil \guessast {}{X | Y} / \card \setZ \right\rceil^\rho}, \label{eq:impGuess}
\ee
where $\guessast {}{\cdot|Y,Z}$ minimizes $\bigEx {}{\guess {}{X|Y,Z}^\rho}$ and $\guessast {}{\cdot|Y}$ minimizes $\bigEx {}{\guess {}{X | Y}^\rho}$. Equality holds whenever $Z = f \! \left( X,Y \right)$ for some $f \colon \setX \times \setY \rightarrow \setZ$ for which $f \! \left( x,y \right) = f \! \left( \tilde x,y \right)$ implies either $\left\lceil G^\ast \! \left( x | y \right) / \card \setZ \right\rceil \neq \left\lceil G^\ast \! \left( \tilde x | y \right) / \card \setZ \right\rceil$ or $x = \tilde x$. Such a function always exists because for all $l \in \naturals$ at most $\card \setZ$ different $x \in \setX$ satisfy $\left\lceil G^\ast \! \left( x | y \right) / \card \setZ \right\rceil = l$.
\end{lemma}

\begin{proof}
If $g \! \left( x, \! y \right) \! \in \! \argmin_{z \in \setZ} G^\ast \! \left( x|y, \! z \right)$ for $(x,y) \in \setX \times \setY$, then $\Ex {}{\guessast {}{X|Y,Z}^\rho} \! \geq \! \min_{G} \Ex {}{\guess {}{X|Y,g \! \left( X,Y \right)}^\rho}$. It thus suffices to prove \eqref{eq:impGuess} for the case where $Z$ is deterministic given $(X,Y)$, and we thus assume w.l.g.\ that $Z \! = \! g \! \left( X,Y \right)$ for some function $g \colon \setX \times \setY \! \rightarrow \! \setZ$. Consider
\ba
\Ex {}{\guess {}{X|Y,Z}^\rho} = \sum_{x,y} P_{X,Y} \! \left( x,y \right) \guess {}{x|y, g \! \left( x, y \right)}^\rho, \label{eq:expToMinImpGuess}
\ea
where $\guess {}{\cdot | Y, g \! \left( X,Y \right)}$ is a guessing function. Note that $\guess {}{x|y,g \! \left( x,y \right)} \! = \! \guess {}{\tilde x|y,g \! \left( \tilde x,y \right)}$ implies $g \! \left( x,y \right) \! \neq \! g \! \left( \tilde x,y \right)$ for all $y \! \in \! \setY$ and distinct $x, \tilde x \in \setX$. For every $l \! \in \! \naturals$ there are thus at most $\card \setZ$ different $x \in \setX$ for which $\guess {}{x|y,g \! \left( x,y \right)} = l$. For each $y \! \in \! \setY$ order the possible realizations of $X$ in decreasing order of $P_{X,Y} \! \left( x,y \right)$, i.e., in decreasing order of their posterior probabilities given $Y \! = \! y$, and let $x_j^y$ denote the $j$-th element. Clearly, \eqref{eq:expToMinImpGuess} is minimum over $g \! \left( \cdot, \cdot \right)$ and $\guess {}{\cdot | Y, g \! \left( X,Y \right)}$ if for $l \! \in \! \naturals$ and $y \! \in \! \setY$ we have $\guess {}{x|y, g \! \left( x, y \right)} = l$ whenever $x = x_j^y$ for some $j$ satisfying $(l - 1) \card \setZ + 1 \leq j \leq l \card \setZ$, i.e., $\left\lceil j / \card \setZ \right\rceil = l$. Since $\guessast {}{\cdot | Y}$ minimizes $\bigEx {}{\guess {}{X | Y}^\rho}$, it orders the elements of $\setX$ in increasing order of their posterior probabilities given $Y$. We can thus choose $x^y_j$ to be the unique $x \! \in \! \setX$ for which $\guessast {}{x | y} = j$. Hence, \eqref{eq:expToMinImpGuess} is minimized if $f \! \left( \cdot, \cdot \right)$ satisfies the specifications in the lemma, $g \! \left( \cdot, \! \cdot \right) \! = \! f \! \left( \cdot, \! \cdot \right)$, and $\guess {}{x|y, \! f \! \left( x, \! y \right)} \! = \! \left\lceil \guessast {}{x|y} / \card \setZ \right\rceil$. The minimum equals the RHS of \eqref{eq:impGuess}.
\end{proof}

Lemma~\ref{le:ImproveGuess} and \cite[Equation~(26)]{buntelapidoth14} imply the following result:

\begin{corollary}\label{co:impGuess}
Draw $\left( X,Y \right)$ from the finite set $\setX \times \setY$ according to the PMF $P_{X,Y}$, and let $\setZ$ be a finite set. There is a function $f \colon \setX \times \setY \rightarrow \setZ$ such that $Z = f \! \left( X,Y \right)$ satisfies
\be 
\min_G \Ex {}{\guess {}{X|Y,Z}^\rho} < 1 + 2^{\rho} \card \setZ^{-\rho} \min_G \Ex {}{\guess {}{X|Y}^\rho}. \label{eq:lbGuessSI}
\ee
Conversely, every chance variable $Z$ with alphabet $\setZ$ satisfies
\be
\min_G \Ex {}{\guess {}{X|Y,Z}^\rho} \geq \card \setZ^{-\rho} \min_G \bigEx {}{\guess {}{X | Y}^\rho} \vee 1. \label{eq:ubGuessSI}
\ee
\end{corollary}

We now sketch the proofs of Theorems~\ref{th:distStorGuess} and \ref{th:distStor} starting with the direct part. Fix $\left( c_s,c_1,c_2 \right) \in \naturals^3$ satisfying \eqref{bl:condCsC1C2Guess} in the ''guessing version'' and \eqref{bl:condCsC1C2} in the ''list version''. For each $\nu \in \left\{ s,1,2 \right\}$ let $V_\nu$ be a chance variable taking value in the set $\setV_\nu = \left\{ 0, \ldots, c_\nu - 1 \right\}$. Corollary~\ref{co:impGuess} and \cite[Proposition~4]{arikan96} imply that there is a $0$-$1$ valued conditional PMF $\distof { \left( V_s, V_1, V_2 \right) = \left( v_s, v_1, v_2 \right) | X = x, Y = y}$ for which
\ba
\!\!\! \min_G \bigEx {}{\guess {}{X|Y, \! V_s, \! V_1, \! V_2}^\rho} \!\! < \! 1 \!\! + \! 2^{\rho \left( \! \renent {\tirho} {X \left| Y \right.} \! - \! \log \! \left( c_s c_1 c_2 \right) \! + \! 1 \! \right)}, \label{eq:BobMomDistStorDir1Guess}
\ea
and likewise \cite[Theorem~6.1]{buntelapidoth14} implies that there is a $0$-$1$ valued conditional PMF for which
\ba
\!\!\! \bigEx {}{|\setL^Y_{V_s,V_1,V_2}|^\rho} \!\! < \! 1 \! + \! 2^{\rho \left( \renent {\tirho} {X | Y} - \log \! \left( c_s c_1 c_2 - \log \! \card \setX - 2 \right) + 2 \right)}. \label{eq:BobMomDistStorDir1}
\ea
Both \eqref{bl:condCsC1C2Guess} and \eqref{bl:condCsC1C2} imply $\card {\setM_1} \geq c_s c_1$ and $\card {\setM_2} \geq c_s c_2$. It thus suffices to prove \eqref{eq:BobMomDistStorDirGuess}--\eqref{eq:EveMomDistStorDirGuess} and \eqref{eq:BobMomDistStorDir}--\eqref{eq:EveMomDistStorDir} for a conditional PMF \eqref{eq:aliceEncPMF} that assigns positive mass only to $c_s c_1$ elements of $\setM_1$ and $c_s c_2$ elements of $\setM_2$, and we thus assume w.l.g.\ that $\setM_1 = \setV_s \times \setV_1$ and $\setM_2 = \setV_s \times \setV_2$. Hence, we can choose $M_1 = \left( V_s \oplus_{c_s} \! U, V_1 \right)$ and $M_2 = \left( U, V_2 \right)$, where $U$ is independent of $\left( X, Y \right)$ and uniformly distributed over $\setV_s$, and where $\oplus_{c_s}$ denotes modulo-$c_s$ addition. For this choice \eqref{eq:BobMomDistStorDirGuess} follows from \eqref{eq:BobMomDistStorDir1Guess} and \eqref{eq:BobMomDistStorDir} from \eqref{eq:BobMomDistStorDir1}. The proof of \eqref{eq:EveMomDistStorDirGuess} and \eqref{eq:EveMomDistStorDir} is more involved. It builds on the following two ideas: 1) Since $U$ is computable from both $(X,M_1)$ and $(X,M_2)$ we can w.l.g.\ assume that Eve must guess $(X,U)$ instead of $X$. 2) Given two guessing functions $\guess 1 { \cdot,\cdot \left| Y, M_1 \right. }$ and $\guess 2 { \cdot,\cdot \left| Y, M_2 \right. }$ for $(X,U)$, one can show that $\guess 1 { \cdot,\cdot \left| Y, M_1 \right. } \wedge \guess 2 { \cdot,\cdot \left| Y, M_2 \right. }$ behaves like a guessing function $\guess {}{\cdot,\cdot | Y,Z}$ for $(X,U)$, where the additional side-information $Z$ assumes at most $c_s \! \left( c_1 + c_2 \right)$ different values. Once 1) and 2) have been established, the proof is concluded by Corollary~\ref{co:impGuess}, \cite[Theorem~1]{arikan96}, and $\renent {\tirho}{X, U | Y } \! = \! \renent {\trho} {X \left| Y \right.} \! + \! \log \! c_s$. 

The converse is straightforward: In the ''guessing version'' the bound \eqref{eq:BobMomDistStorConvGuess} on Bob's ambiguity follows from Corollary~\ref{co:impGuess} and \cite[Theorem~1]{arikan96}. In the ''list version'' \eqref{eq:BobMomDistStorConv} follows from \cite[Theorem~6.1]{buntelapidoth14} and the observation that $\mathscr A_{\text{B}}^{(\text g)}$ is minimized if the PMF in \eqref{eq:aliceEncPMF} is $0$-$1$ valued. Clearly, Eve's ambiguity satisfies $\mathscr A_{\text{E}} \! \left( P_{X,Y} \right) \leq \min_{k \in \left\{1,2\right\}} \Bigl( \min_{G_k} \Ex {}{\guess {k}{X|Y,M_k}^\rho} \Bigr)$, and Corollary~\ref{co:impGuess} implies for each $k \! \in \! \left\{1,2\right\}$ and $l \! \in \! \left\{ 1, \! 2 \right\} \! \setminus \! \set k$
\be 
\min_G \Ex {}{\guess {}{X|Y,M_1,M_2}^\rho} \geq \card {\setM_l}^{-\rho} \min_{G_k} \Ex {}{\guess k{X|Y,M_k}^\rho}.
\een
Since $\min_G \Ex {}{\guess {}{X|Y,M_1,M_2}^\rho} \leq \Ex {}{ | \setL^Y_{M_1,M_2} |^\rho}$ we thus find that in both versions Eve's ambiguity exceeds Bob's by at most a factor of $\card {\setM_1}^\rho \wedge \card {\setM_2}^\rho$. Due to \cite[Proposition~4]{arikan96} and because Eve can guess $X$ using only $Y$ we have $\min_G \Ex {}{\guess {}{X|Y}^\rho} \leq 2^{\rho \renent {\tirho}{X|Y}}$. Hence, \eqref{eq:EveMomDistStorConvGuess} and \eqref{eq:EveMomDistStorConv} hold.

\section{Discussion}\label{sec:guessListClose}

We next explain why the results for the ''guessing version'' and the ''list version'' differ only by polylogarithmic factors of $\card \setX$. In the ''guessing version'' Bob uses a guessing function $\guess {}{\cdot | Y,M_1,M_2}$ to guess $X$ based on the side-information $Y$ and the hints $M_1$ and $M_2$, and his ambiguity is $\Ex {}{\guess {}{X | Y,M_1,M_2}^\rho}$. On account of Corollary~\ref{co:guessToList} we can construct an additional hint $M$, which assumes $\log \! \card \setX + 1$ different values and satisfies $\Ex {}{ | \setL^Y_{M_1,M_2,M}|^\rho} \leq \bigEx {}{\guessast {}{X | Y, M_1, M_2 }^\rho}$, where $\setL^Y_{M_1,M_2,M}$ is the smallest list that is guaranteed to contain $X$ given $(Y,M_1,M_2,M)$. Suppose Alice maps $X$ to the hints $M_1' = (M_1,M)$ and $M_2' = M_2$. Then, Bob's ambiguity in the ''list version'' is $\bigEx {}{| \setL^Y_{M_1',M_2'}|^\rho} = \Ex {}{ | \setL^Y_{M_1,M_2,M}|^\rho}$ and hence no larger than $\Ex {}{\guess {}{X | Y,M_1,M_2}^\rho}$. Since $M$ assumes only $\log \! \card \setX + 1$ different values one can use Lemma~\ref{le:ImproveGuess} to show that Eve's ambiguity decreases by at most a polylogarithmic factor of $\card \setX$ (compared to the case where the hints are $M_1$ and $M_2$).

\section{Extensions}

In the full version of this paper \cite{bracherlapidoth14}, we discuss several modifications of the model: We extend the analysis to scenarios where we know in advance which hint Eve will find (e.g., since the other hint is stored in a safe), or where secrecy is achieved by means of a secret key, which is available to Alice and Bob but not to Eve. Moreover, we generalize the asymptotic results to the case where Bob and Eve must reconstruct $X^n$ within a given distortion $D$ (cf.\ \cite{arikanmerhav98,buntelapidoth14}).

\end{document}